\documentclass[journal]{IEEEtran}
\usepackage{amsmath}
\usepackage{epsfig}
\usepackage{latexsym}
\usepackage{mathrsfs}
\usepackage{amsthm}
\usepackage{amssymb}
\usepackage{amsfonts}
\usepackage{amsbsy}
\usepackage[shortlabels]{enumitem}
\usepackage{url}
\usepackage{multirow}
\usepackage{array}
\usepackage{pdflscape}
\usepackage{xcolor}
\usepackage{stmaryrd}
\usepackage{verbatim}
\usepackage{caption}

\renewcommand{\arraystretch}{1.1}
\theoremstyle{plain}
\newtheorem{thm}{Theorem}
\newtheorem{lem}[thm]{Lemma}
\newtheorem{prop}[thm]{Proposition}
\newtheorem{cor}{Corollary}
\theoremstyle{definition}
\newtheorem{defn}{Definition}
\newtheorem{ex}{Example}
\theoremstyle{remark}
\newtheorem{rem}{Remark}
\newcommand{\Fq}{\mathbb{F}_{q}}

\newcommand{\To}{\longrightarrow}
\newcommand{\E}{\mathbb{E}}
\newcommand{\F}{{\mathbb F}}

\newcommand{\K}{\mathbb{K}}
\newcommand{\LL}{\mathbb{L}}
\newcommand{\N}{\mathbb{N}}
\newcommand{\Tr}{{\rm Tr}}

\newcommand{\Span}{{\rm Span}}

\begin{document}

\title{A Comparison of Distance Bounds for Quasi-Twisted Codes}

\author{\centerline{Martianus Frederic Ezerman, John Mark Lampos, San Ling, Buket \"{O}zkaya, and Jareena Tharnnukhroh}
\thanks{M. F. Ezerman, S. Ling, B. \"{O}zkaya, and J. Tharnnukhroh are with the School of Physical and Mathematical Sciences, Nanyang Technological University, 21 Nanyang Link, Singapore 637371, e-mails: $\{\rm fredezerman, lingsan, buketozkaya\}$@ntu.edu.sg, ${\rm jareena001}$@e.ntu.edu.sg.}
\thanks{J.~M.~Lampos is with the Institute of Mathematical Sciences and Physics, University of the Philippines, Los Ba\~{n}os, Laguna, Philippines 4031, 
e-mail: ${\rm jtlampos}$@up.edu.ph.}
\thanks{M. F. Ezerman, S. Ling, and B. \"{O}zkaya are supported by Nanyang Technological University Research Grant No. 04INS000047C230GRT01.}
\thanks{J.~M.~Lampos is supported by DOST-ASTHRDP Dissertation Grant and CHED K-12 Transition Program Scholarship for Graduate Studies Abroad.}
\thanks{J.~Tharnnukhroh's scholarship is from the Development and Promotion of Science and Technology (DPST) talent project of Thailand.}
\thanks{The fourth section of this paper contains results presented at the 2019 IEEE International Symposium on Information Theory \cite{ELOT}.}
\thanks{This work has been submitted to the IEEE for possible publication. Copyright may be transferred without notice, after which this version may no longer be accessible.}
}


\maketitle

\begin{abstract}
Spectral bounds on the minimum distance of quasi-twisted codes over finite fields are proposed, based on eigenvalues of polynomial matrices and the corresponding eigenspaces. They generalize the Semenov-Trifonov and Zeh-Ling bounds in a way similar to how the Roos and shift bounds extend the BCH and HT bounds for cyclic codes. The eigencodes of a quasi-twisted code in the spectral theory and the outer codes in its concatenated structure are related. A comparison based on this relation verifies that the Jensen bound always outperforms the spectral bound under special conditions, which yields a similar relation between the Lally and the spectral bounds. The performances of the Lally, Jensen and spectral bounds are presented in comparison with each other.
\end{abstract}

\begin{IEEEkeywords}
quasi-twisted code, concatenated code, minimum distance bound, polynomial matrices, spectral analysis
\end{IEEEkeywords}

\section{Introduction}
\label{intro}

Cyclic codes have been widely studied, since their algebraic structure provides effective encoding and decoding. Several lower bounds on the minimum distance of cyclic codes had been derived. The first and perhaps the most famous one was obtained by Bose and Chaudhuri (\cite{BC}) and by Hocquenghem (\cite{H}), known as the BCH bound. An extension of the BCH bound was formulated by Hartmann and Tzeng in \cite{HT}, which can be considered as a two-directional BCH bound. The Roos bound (\cite{R2}) generalized this idea further by allowing the HT bound to have a certain number of gaps in both directions, which was extended to constacyclic codes in \cite{RZ}. Another remarkable extension of the HT bound, known as the shift bound, was introduced by van Lint and Wilson in \cite{LW}. The shift bound is known to be particularly powerful on many nonbinary codes ({\it e.g.} see \cite{EL}). 

Quasi-twisted (QT) codes form an important class of block codes that includes cyclic codes, quasi-cyclic (QC) codes and constacyclic codes as special subclasses. In addition to their rich algebraic structure (\cite{Y,SZ}), QT codes are also known to be asymptotically good (\cite{C,DH,WS}) and they yield good parameters (\cite{AGOSS1,AGOSS2,QSS,SQS}).

Even though QC and QT codes are interesting from both theoretical and practical points of view, the study on their minimum distance estimates is not as rich as for cyclic and constacyclic codes. Jensen derived a significant bound in \cite{J}, which is valid for many code classes having a concatenated structure, including QT codes (\cite{LG}). Lally gave another estimate on the minimum distance of a given QC code, which is obtained by a simpler concatenation \cite{L2}, depending only on the index and the co-index of the QC code in consideration. More recently, Semenov and Trifonov developed a spectral analysis of QC codes (\cite{ST}), based on the work of Lally and Fitzpatrick in \cite{LF}, and formulated a BCH-like minimum distance bound, together with a comparison with a few other bounds for QC codes. Their approach was generalized by Zeh and Ling in \cite{LZ}, by using the HT bound. They extended the spectral method to QC product codes in \cite{LZ2}. The first spectral analysis of QT codes appeared in \cite{ELOT}, where a spectral bound under some restricted conditions was proven and then this bound was only compared with the BCH-like and HT-like versions.

In this work, we investigate the spectral theory for QT codes, by following the steps in \cite{ELOT,ST}, that is centered around the eigenvalues of a given QT code. They can be considered as the QT analogues of the zeros of constacyclic codes. We aim at deriving a general spectral bound which holds for a larger choice of eigenvalues than the more constrained version in \cite{ELOT}. For this, we focus on connecting the concatenated structure of QT codes to the key elements of the spectral method. Using this relation, we prove three important results. First, we push the general spectral bound to the largest possible extent and show that it holds for any nonempty subset of eigenvalues. Second, we establish a theoretical link between the spectral and Jensen bounds. We show that deploying all of the eigenvalues causes the Jensen bound to win against the spectral bound, but looking at proper subsets of the eigenvalues may occasionally turn the situation to the opposite, allowing the spectral bound to beat the Jensen bound. At last, we prove a relation between the Lally and spectral bounds, which mimics the result when comparing the spectral bound with the Jensen bound, except that this time the largest possible set of eigenvalues is considered. The numerical comparisons at the end present how all these three bounds behave over a large number of randomly chosen QT codes.

This paper is organized as follows. Section \ref{basics} recalls some required background on the algebraic structure of constacyclic and QT codes. Section \ref{comparison section} studies the spectral method in terms of the concatenated structure of QT codes. Using this unified approach, Section \ref{bound sect} modifies the spectral method of Semenov-Trifonov to QT codes. A generalized spectral bound on the minimum distance of QT codes is formulated and proven, where the Roos-like and shift-like bounds for QT codes follow as special cases. In Section \ref{res sect}, we first prove that the Jensen bound performs at least as well as the spectral bound under certain assumptions. Then we formulate the Lally bound for QT codes and compare it with the spectral bound in a similar way. Section \ref{Example section} presents explicit examples and constructions over a large set of random QT codes, displaying the performances of these three bounds in terms of sharpness and rank. All computations were carried out in \textsc{magma} \cite{BCP}.
 
\section{Background}\label{basics}

\subsection{Constacyclic codes and minimum distance bounds from defining sets} \label{consta sect}

Let $\Fq$ denote the finite field with $q$ elements, where $q$ is a prime power. Let $m$ be throughout a positive integer with $\gcd(m,q)=1$. For a fixed $\lambda \in \Fq\setminus \{0\}$, a linear code $C\subseteq \Fq^m$ is called a $\lambda$-{\it constacyclic} code if it is invariant under the $\lambda$-constashift of codewords, {\it i.e.}, $(c_0,\ldots,c_{m-1}) \in C$ implies $(\lambda c_{m-1},c_0,\ldots,c_{m-2}) \in C$. In particular, if $\lambda = 1$ or $q=2$, then $C$ is a cyclic code.

Consider the principal ideal $I=\langle x^m -\lambda \rangle$ of $\Fq[x]$ and define the residue class ring $R:=\Fq[x]/I$. For an element $\mathbf{a}\in \Fq^m$, we associate an element of $R$ via the following $\Fq$-module isomorphism:
\begin{equation}\begin{array}{rcl} \label{identification-1}
\phi: \ \F_q^{m} & \longrightarrow & R  \\
\mathbf{a}=(a_0,\ldots,a_{m-1}) & \longmapsto & a(x):= a_0+\cdots + a_{m-1}x^{m-1}.
\end{array}
\end{equation}
While elements in $R$ are cosets of the form $b(x) + I$ where $b(x) \in \F_q[x]$, we write them with a slight abuse of notation as $b(x)$. Observe that the $\lambda$-constashift in $\F_q^{m}$ corresponds to multiplication by $x$ in $R$. Therefore, a $\lambda$-constacyclic code $C\subseteq \Fq^m$ can be viewed as an ideal of $R$. Since every ideal in $R$ is principal, there exists a unique monic polynomial $g(x)\in R$ such that $C=\langle g(x)\rangle$, {\it i.e.}, each codeword $c(x)\in C$ is of the form $c(x)=a(x)g(x)$, for some $a(x)\in R$. The polynomial $g(x)$, which is a divisor of $x^m-\lambda$, is called the {\it generator polynomial} of $C$, whereas the {\it check polynomial} of $C$, say $h(x)\in R$, satisfies $g(x)h(x)=x^m-\lambda$. 

Let $\mbox{wt}(c)$ denote the number of nonzero coefficients in $c(x)\in C$. Recall that the minimum distance of $C$ is defined as $d(C):=\min\{\mbox{wt}(c) : 0\neq c(x)\in C\}$ when $C$ is not the trivial zero code. For any positive integer $p$, let $\mathbf{0}_p$ denote throughout the all-zero vector of length $p$. We have $C=\{\mathbf{0}_m\}$ if and only if $g(x)=x^m-\lambda$. In this case, we assume throughout that $d(C)=\infty$.

Let $r$ be the smallest divisor of $q-1$ with $\lambda^r=1$ and let $\alpha$ be a primitive $rm^{\rm th}$ root of unity such that $\alpha^m=\lambda$. Then, $\xi:=\alpha^r$ is a primitive $m^{\rm th}$ root of unity and the roots of $x^m-\lambda$ are of the form $\alpha, \alpha\xi, \ldots, \alpha\xi^{m-1}$. Henceforth, let $\Omega :=\{\alpha\xi^k : 0\leq k \leq m-1\}=\{\alpha^{1+kr} : 0\leq k \leq m-1\}$ be the set of all $m^{\rm th}$ roots of $\lambda$ and let $\F$ be the smallest extension of $\Fq$ that contains $\Omega$ (equivalently, $\F=\Fq(\alpha)$ so that $\F$ is the splitting field of $x^m-\lambda$). Given the $\lambda$-constacyclic code $C=\langle g(x)\rangle$, the set of roots of its generator polynomial, say $$L:=\{\alpha\xi^k\ :\ g(\alpha\xi^k)=0\}\subseteq \Omega,$$ is called the {\it zero set} of $C$. The power set $\mathcal{P}(L)$ of $L$ is called the {\it defining set} of $C$. Clearly, $L=\emptyset$ if and only if $C=\langle 1\rangle = \Fq^m$. Note that $\alpha\xi^k\in L$ implies $\alpha^q\xi^{qk}\in L$, for each $k$, where $\alpha^q\xi^{qk}=\alpha\xi^{k'}$ with $k'=\frac{q-1}{r}+qk \mod m$. A nonempty subset $E\subseteq\Omega$ is said to be {\it consecutive} if there exist integers $e,n$ and $\delta$ with $e\geq 0,\delta \geq 2, n> 0$ and $\gcd(m,n)=1$ such that
\begin{equation} \label{cons zero set}
E:=\{\alpha\xi^{e+zn}\ :\ 0\leq z\leq \delta-2\}\subseteq\Omega.
\end{equation}

Let $\mathcal{P}(\Omega)$ denote the power set of $\Omega$. Observe that any $P\in\mathcal{P}(\Omega)$ is the zero set of some $\lambda$-constacyclic code $D_P\subseteq\F^m$ since $x^m-\lambda$ splits into linear factors over $\F$.
Let $C$ be a nontrivial $\lambda$-constacyclic code of length $m$ over some subfield of $\F$ with zero set $L\subseteq \Omega$. Then, for any $P\subseteq L$, $C$ is contained in $D_P$ and therefore we have $d(C)\geq d(D_P)$. 
Throughout, we define a {\it defining set bound} to be a member of a chosen family $\mathcal{B}(C):=\{(P, d_P)\} \subseteq \mathcal{P}(\Omega) \times (\N \cup \{\infty\})$ such that, for any $(P, d_P)\in \mathcal{B}(C)$, $P\subseteq L$ implies $d(C) \geq d(D_P) \geq d_P$ (a more detailed formulation given for the cyclic codes can be found in \cite{BGS}). We set 
\begin{multline*}\mathcal{B}_1(C) := \{(P, d(D_P)) : \\
D_P\subseteq\F^m \mbox{\ has\ zero\ set\ } P, \mbox{\ for\ all\ } P\subseteq L\}.
\end{multline*}
In particular, following the notation given above in the case when $P=L=\Omega$, we have $D_{\Omega}=\{\mathbf{0}_m\}$ over $\F$ with $d(D_{\Omega})=\infty$ and consequently, we include $(\Omega,\infty)$ in every collection $\mathcal{B}(C)$ as a convention when $L=\Omega$.

If we choose $$\mathcal{B}_2(C) := \{(E, |E|+1) : E\subseteq L \mbox{\ is\ consecutive}\},$$ then we obtain the BCH bound, where $E$ is of the form given in (\ref{cons zero set}). Similarly, we formulate the HT bound as \begin{multline*}\mathcal{B}_3(C) := \{(D, \delta + s) : \\
D = \{\alpha\xi^{e + z  n_1 + y  n_2} : 0\leq z \leq \delta-2, 0 \leq y\leq s\} \subseteq L\},\end{multline*} for integers $ e \geq 0$, $\delta \geq 2$ and positive integers $s, n_1$ and $n_2$ such that $\gcd(m, n_1) = 1$ and $\gcd(m, n_2) < \delta$. Note that any union of defining set bounds is again a defining set bound.

We are now ready to present the Roos bound on the minimum distance of a given $\lambda$-constacyclic code (see \cite[Theorem 2]{R2} for the original version by C. Roos for cyclic codes). For the proof of the result below, we refer to \cite[Theorem 6]{RZ}.

\begin{thm}[Roos bound] \label{Roos}
Let $N$ and $M$ be two nonempty subsets of $\Omega$. If there exists a consecutive set $M'$ containing $M$ such that $|M'| \leq |M| + d_N -2$, then we have $d_{MN}\geq |M| + d_N -1$ where $MN:=\displaystyle\frac{1}{\alpha}\bigcup_{\varepsilon\in M} \varepsilon N$.
\end{thm}
If $N$ is consecutive like in (\ref{cons zero set}), then we obtain the following.
\begin{cor}\cite[Corollary 1]{RZ},\cite[Corollary 1]{R2} \label{Roos2}
Let $N, M$ and $M'$ be as in Theorem \ref{Roos}, with $N$ consecutive. Then $|M'| < |M| + |N|$ implies $d_{MN}\geq |M| + |N|$.
\end{cor}

\begin{rem}\label{Roos remark}
In particular, the case $M=\{\alpha\}$ yields the BCH bound for the associated constacyclic code (see \cite[Corollary 2]{RZ}). The original BCH bound for cyclic codes can be found in \cite{BC} and \cite{H}. By taking $M'=M$, we obtain the HT bound (see \cite[Corollary 3]{RZ}) and the HT bound for cyclic codes is given in \cite[Theorem 2]{HT}. 
\end{rem}

Another improvement to the HT bound for cyclic codes was provided by van Lint and Wilson in \cite{LW}, which is known as the shift bound. We proceed by formulating the shift bound for constacyclic codes. To do this, we need the notion of an {\it independent set}, which can be constructed over any field in a recursive way, as given below.

Let $S$ be a subset of some field $\K$ of any characteristic. One inductively defines a family of finite subsets of $\K$, called independent with respect to $S$, as follows.
\begin{enumerate}
\item $\emptyset$ is independent with respect to $S$.
\item If $A\subseteq S$ is independent with respect to $S$, then $A\cup\{b\}$ is independent with respect to $S$, for any $b\in \K \setminus S$.
\item If $A$ is independent with respect to $S$ and $c$ is any nonzero element in $\K$ such that $cA\subseteq S$, then $cA$ is independent with respect to $S$.
\end{enumerate}

The recursive construction starts with the smallest independent set, say $A_0 = \emptyset$. Then $A_1 = \emptyset \cup\{b_0\} = \{b_0\}$, for some $b_0 \in  \K\setminus S$. One can continue with $A_2 = c_1 A_1 \cup \{b_1\}$ provided that $c_1 A_1 \subseteq S$ for some $c_1\in \K\setminus \{0\}$ and $b_1\in \K\setminus S $. This ensures that $A_2$ is again independent with respect to $S$. The process stops at the $t^{\rm th}$ step when there is no constant $c\in \K\setminus \{0\}$ such that $c A_t \subseteq S$. Observe that $|A_0|=0$ and $|A_{i+1}|=|A_i|+1$ for all $ i \in \{0,\ldots,t-1\}$.

\begin{thm}[Shift bound] \cite[Theorem 11]{LW}\label{shift bound}
Let $0\neq f(x)\in \K[x]$ and $S=\{\theta\in \K\  : \ f(\theta)=0\}$. Then \rm{wt}$(f)\geq |A|,$ for every subset $A$ of $\K$ that is independent with respect to $S$.
\end{thm}

The shift bound for a given $\lambda$-constacyclic code follows by considering the weights of its codewords $c(x)\in C$ and the independent sets with respect to subsets of its zero set $L$. Observe that, in this case, the universe of the independent sets is $\Omega$, not the extension field $\F$, because all of the possible roots of the codewords are contained in $\Omega$. Moreover, we choose $b$ from $\Omega \setminus P$ in Condition (2) above, where $P\subseteq L$, and $c$ in Condition (3) is of the form $\xi^k\in\F\setminus \{0\}$, for some $0\leq k\leq m-1$. 

\begin{cor}\label{shift remark}
The BCH and the HT bounds for a given $\lambda$-constacyclic code $C$ can be obtained from the shift bound as follows:
	\begin{enumerate}
	\item[i.] The set $B_{\delta} :=\{\alpha\xi^{e+zn} : 0\leq z \leq\delta-1\}$ is independent with respect to the consecutive set $E$ in (\ref{cons zero set}) and $d(C_E)\geq |B_{\delta}|=\delta$.\vspace{5pt}
	\item[ii.] Let $D = \{\alpha\xi^{e + z  n_1 + y  n_2} : 0\leq z \leq \delta-2, 0\leq y\leq s\}$, for integers $ e \geq 0$, $\delta \geq 2$ and positive integers $s, n_1$ and $n_2$ such that $\gcd(m, n_1) = 1$ and $\gcd(m, n_2) < \delta$. Then, for any fixed $\zeta\in \{0,\ldots,\delta-2\}$, the set 
\begin{multline*}
A_{\zeta}:=\{\alpha\xi^{e + z  n_1} :	0\leq z \leq\delta-2\}\cup \\ \{\alpha\xi^{e + {\zeta}  n_1 + y  n_2} : 0\leq y \leq s+1\}
\end{multline*} 
is independent with respect to $D$ and $d(C_D)\geq\delta+s$.
	\end{enumerate}
\end{cor}

\begin{proof}
\begin{enumerate}
\item[i.] We construct a sequence $B_0, B_1, \ldots, B_{\delta}\subseteq\Omega$ of independent sets with respect to $E$ as follows. 
\begin{eqnarray*}
B_0&=&\emptyset,\\
B_1&=&B_0 \cup \{\alpha\xi^{e+(\delta-1)n}\}=\{\alpha\xi^{e+(\delta-1)n}\},\\
B_2&=&\xi^{-n} B_1 \cup \{\alpha\xi^{e+(\delta-1)n}\}\\
&=&\{\alpha\xi^{e+(\delta-2)n},\alpha\xi^{e+(\delta-1)n}\},\\
&\vdots&\\
B_{\delta}&=&\xi^{-n} B_{\delta-1} \cup \{\alpha\xi^{e+(\delta-1)n}\}\\
&=&\{\alpha\xi^{e}, \alpha\xi^{e+n}, \ldots, \alpha\xi^{e+(\delta-1)n}\}.
\end{eqnarray*}
Since there is no element $\xi^{k}\in\F\setminus \{0\}$ such that $\xi^k B_{\delta}\subseteq E$, for all $k\in \{0,\ldots,m-1\}$, the process stops. By Theorem \ref{shift bound}, $|B_{\delta}|=\delta$ implies $d(C_E)\geq \delta$.\vspace{10pt}

\item[ii.] Let $a:=\alpha\xi^{e+(\delta-1)n_1+sn_2}$ and $b_{\zeta}:=\alpha\xi^{e+{\zeta}n_1+(s+1)n_2}$, for some fixed $\zeta\in \{0,\ldots,\delta-2\}$. Note that $a,b_{\zeta}\in\Omega\setminus D$ and consider the following sequence of independent sets with respect to $D$.
\begin{eqnarray*}
A_0&=&\emptyset,\\
A_1&=&A_0 \cup \{a\}=\{\alpha\xi^{e+(\delta-1)n_1+sn_2}\},\\
A_2&=&\xi^{-n_1}A_1 \cup \{a\}\\
&=&\{\alpha\xi^{e+(\delta-2)n_1+sn_2}, \alpha\xi^{e+(\delta-1)n_1+sn_2}\},\\
&\vdots&\\
A_{\delta-1}&=&\xi^{-n_1}A_{\delta-2} \cup \{a\}\\
&=&\{\alpha\xi^{e+n_1+sn_2}, \alpha\xi^{e+2n_1+sn_2}, \ldots,\\ &&\hspace{105pt} \alpha\xi^{e+(\delta-1)n_1+sn_2}\},\\
A_{\delta}&=&\xi^{-n_1}A_{\delta-1} \cup \{b_{\zeta}\}\\
&=&\{\alpha\xi^{e+sn_2}, \alpha\xi^{e+n_1+sn_2}, \ldots, \\ &&\hspace{30pt} \alpha\xi^{e+(\delta-2)n_1+sn_2}, \alpha\xi^{e+\zeta n_1+(s+1)n_2}\},\\
A_{\delta+1}&=&\xi^{-n_2}A_{\delta} \cup \{b_{\zeta}\}\\
&=&\{\alpha\xi^{e+(s-1)n_2},\ldots, \alpha\xi^{e+(\delta-2)n_1+(s-1)n_2},\\ &&\hspace{45pt}  \alpha\xi^{e+{\zeta} n_1+sn_2}, \alpha\xi^{e+{\zeta} n_1+(s+1)n_2}\},\\
&\vdots&\\		
A_{\delta+s}&=&\xi^{-n_2}A_{\delta+s-1} \cup \{b_{\zeta}\}=A_{\zeta}.
\end{eqnarray*}
\normalsize
We have no $\xi^{k}\in\F\setminus \{0\}$ such that $\xi^kA_z\subseteq D$, for all $k\in \{0,\ldots,m-1\}$. Thus, we obtain $d(C)\geq |A_{\zeta}|=\delta+s$.\qedhere\hspace{20pt}
\end{enumerate}	
\end{proof}

\begin{rem}\label{Roos-shift rem}
Let $C$ be a nontrivial $\lambda$-constacyclic code of length $m$ over some subfield of $\F$ with zero set $L\subseteq \Omega$. Then, the Roos bound corresponds to the choice \begin{multline*}\mathcal{B}_4(C):=\{(MN,|M|+d_N-1) :  \mbox{there\ exists\ a\ consecutive}\\ \mbox{set\ } M'\subseteq\Omega\  \mbox{such\ that\ }  M'\supseteq M  \mbox{\ with\ } |M'|\leq |M|+d_N-2\},
\end{multline*} for any $\emptyset\neq MN\subseteq L$ with $MN=\frac{1}{\alpha}\bigcup_{\varepsilon\in M} \varepsilon N$. On the other hand, if we pick \begin{multline*}\mathcal{B}_5(C):=\{(T_A,|A|) : A\subseteq\Omega \mbox{\ is\ independent}\\ \mbox{ with\ respect\ to\ } L,\ T_A=A\cap L\},
\end{multline*} then we obtain the shift bound.
\end{rem}

\subsection{Quasi-twisted codes and their concatenated structure} \label{QT sect}

We assume the notation above and let $\ell$ be a positive integer. A linear code $C\subseteq\Fq^{m\ell}$ is called a $\lambda$-quasi-twisted ($\lambda$-QT) code of index $\ell$ and co-index $m$ if it is invariant under the $\lambda$-constashift of codewords by $\ell$ positions and $\ell$ is the least positive integer with this property. In particular, if $\ell=1$, then $C$ is a $\lambda$-constacyclic code, and if $\lambda = 1$ or $q=2$, then $C$ is a QC code of index $\ell$. If we view a codeword $\mathbf{c}\in C$ as an $m \times \ell$ array
\begin{equation}\label{array}
\mathbf{c}=\left(
  \begin{array}{ccc}
    c_{00} & \ldots & c_{0,\ell-1} \\
    \vdots &  & \vdots \\
    c_{m-1,0} & \ldots & c_{m-1,\ell-1} \\
  \end{array}
\right),\end{equation} then being invariant under $\lambda$-constashift by $\ell$ positions in $\Fq^{m\ell}$ corresponds to being closed under row $\lambda$-constashift in $\Fq^{m\times\ell}$.

If $T_{\lambda}$ denotes the $\lambda$-constashift operator on $\Fq^{m\ell}$, we denote its action on $\mathbf{v} \in \F_q^{m\ell}$ by $T_{\lambda}\cdot \mathbf{v}$. Then $\Fq^{m\ell}$ has an $\Fq[x]$-module structure given by the multiplication
\vspace{-5pt}
 \begin{eqnarray*}
 \Fq[x] \times \F_{q}^{m\ell} & \longrightarrow & \ \F_{q}^{m\ell}\\
 (a(x),\mathbf{v})\ & \longmapsto & a(T_{\lambda}^\ell)\cdot \mathbf{v} .
\end{eqnarray*}
Note that, for $a(x)=x^m-\lambda$, we have $a(T_{\lambda}^\ell)\cdot \mathbf{v} =
 (T_{\lambda}^{m\ell})\cdot \mathbf{v} - \lambda \mathbf{v}=0.$ Hence, a multiplication by elements of $R$ is induced on $\F_q^{m\ell}$ and it can be viewed as an $R$-module. Therefore, a $\lambda$-QT code $C\subseteq \F_q^{m\ell}$ of index $\ell$ is an $R$-submodule of $\F_q^{m\ell}$. 

For an element $\mathbf{c}\in \Fq^{m\times \ell} \simeq \Fq^{m\ell}$, which is represented as in (\ref{array}), we associate an element of $R^\ell$ (cf. (\ref{identification-1}))
\begin{equation} \label{associate-1}
\mathbf{c}(x):=(c_0(x),c_1(x),\ldots ,c_{\ell-1}(x)) \in R^\ell ,
\end{equation}
where, for each $0\leq j \leq \ell-1$, 
\begin{equation}\label{columns} 
c_j(x):= c_{0,j}+c_{1,j}x+c_{2,j}x^2+\cdots + c_{m-1,j}x^{m-1} \in R .
\end{equation} 
The isomorphism $\phi$ in (\ref{identification-1}) extends naturally to
\begin{equation}
\begin{array}{lll} \label{identification-2}
\Phi: \hspace{2cm} \F_q^{m\ell} & \longrightarrow & R^\ell  \\
\mathbf{c}=\left(
  \begin{array}{ccc}
    c_{00} & \ldots & c_{0,\ell-1} \\
    \vdots &  & \vdots \\
    c_{m-1,0} & \ldots & c_{m-1,\ell-1} \\
  \end{array}
\right) & \longmapsto &  \begin{array}{lcl}\vspace{14pt} \\ \hspace{-2pt}\mathbf{c}(x)\\
      \hspace{4pt}\shortparallel \\
    \hspace{-30pt} (c_0(x), \ldots, c_{\ell-1}(x))\\
  \end{array}\\
\hspace{1.2cm}\begin{array}{ccccc}
      \downarrow &  & \downarrow & & \\
     c_0(x) & \ldots & c_{\ell-1}(x) & &\\
  \end{array} & &
\end{array}\\
\end{equation}
Observe that the row $\lambda$-constashift invariance in $\Fq^{m\times\ell}$ corresponds to being closed under componentwise multiplication by $x$ in $R^\ell$. Therefore, the map $\Phi$ above yields an $R$-module isomorphism and any $\lambda$-QT code $C\subseteq \F_q^{m\ell}\simeq \Fq^{m\times\ell}$ of index $\ell$ can be viewed as an $R$-submodule of $R^\ell$.

We now describe the decomposition of a $\lambda$-QT code over $\F_q$ into shorter codes over (field) extensions of $\F_q$. We refer the reader to \cite{Y} for the respective proofs of the following assertions and for the treatment that includes the general repeated-root case ({\it i.e.}, when $\gcd(m,q)\geq 1$). We assume that $x^m-\lambda$ factors into irreducible polynomials in $\F_q[x]$ as
\begin{equation}\label{irreducibles}
x^m-\lambda=f_1(x)f_2(x)\cdots f_s(x).
\end{equation}
Since $m$ is relatively prime to $q$, there are no repeating factors in (\ref{irreducibles}). By the Chinese Remainder Theorem (CRT), we have the following ring isomorphism
\begin{equation} \label{CRT-1}
R\cong \bigoplus_{i=1}^{s} \F_q[x]/\langle f_i(x)\rangle .
\end{equation}
For each $i\in\{1,2,\ldots,s\}$, let $u_i$ be the smallest nonnegative integer such that $f_i(\alpha\xi^{u_i})=0$.
The $\F_{q}$-conjugacy class (or the $q$-cyclotomic class) containing $\alpha\xi^{u_i}$ in $\Omega$ is defined as
\begin{align}\label{cyc coset}
\big[\alpha\xi^{u_i}\big] &= \left\{ \alpha\xi^{u_i},\alpha^q\xi^{qu_i},\alpha^{q^2}\xi^{q^2u_i},\ldots,\alpha^{q^{e_i-1}}\xi^{q^{e_i-1}u_i}\right\} \notag \\ 
&\subseteq \Omega,
\end{align}
where $e_i=\mbox{deg}(f_i)$ and therefore $\big[\alpha\xi^{u_i}\big]$ contains all roots of the irreducible polynomial $f_i$, for each $i$. Note that $\Omega$ is a disjoint union of such $\Fq$-conjugacy classes.

Since the $f_i(x)$'s are irreducible, the direct summands in (\ref{CRT-1}) can be viewed as field extensions of $\F_q$, obtained by adjoining the element $\alpha\xi^{u_i}$. If we set $\E_i:=\Fq(\alpha\xi^{u_i})\cong \Fq[x]/\langle f_i(x) \rangle$, for each $1\leq i \leq s$, then $\E_i$ is an intermediate field between $\F$ and $\Fq$ such that $\big[\E_i : \F_q\big]=e_i$ and we have (cf. (\ref{CRT-1}))
\begin{eqnarray} \label{CRT-2}
R & \simeq & \E_1 \oplus \cdots \oplus \E_{s}
 \nonumber \\
a(x) & \mapsto & \left(a(\alpha\xi^{u_1}),\ldots ,a(\alpha\xi^{u_s}) \right).
\end{eqnarray}
This implies that
\begin{eqnarray} \label{CRT-3}
R^{\ell}&\simeq & \E_1^{\ell} \oplus \cdots  \oplus \E_{s}^{\ell}
\nonumber \\ 
\mathbf{a}(x) & \mapsto & \left(\mathbf{a}(\alpha\xi^{u_1}),\ldots ,\mathbf{a}(\alpha\xi^{u_s}) \right),
\end{eqnarray}
where $\mathbf{a}(\delta)$ denotes the componentwise evaluation at $\delta\in\F$, for any $\mathbf{a}(x)=\bigl(a_{0}(x),\ldots ,a_{\ell-1}(x)\bigr)\in R^{\ell}$. Hence, a $\lambda$-QT code $C\subseteq R^\ell$ can be viewed as an $(\E_1 \oplus \cdots \oplus \E_{s})$-submodule of $\E_1^{\ell} \oplus
\cdots  \oplus \E_{s}^{\ell}$ and it decomposes as
\begin{equation} \label{constituents}
C \simeq C_1\oplus \cdots  \oplus C_{s},
\end{equation}
where $C_i$ is a linear code in $\E_i^{\ell}$, for each $i$. These linear codes over various extensions of $\F_q$ are called the {\it constituents} of $C$ (see \cite[\S 7]{Y} for explicit examples).

Let $C\subseteq R^\ell$ be generated as an $R$-module by
\[
\left\{\bigl(a_{1,0}(x),\ldots ,a_{1,\ell-1}(x)\bigr),\ldots,\bigl(a_{r,0}(x),\ldots ,a_{r,\ell-1}(x)\bigr)\right\}.
\]
Then, for $1\leq i \leq s$, we have
\begin{equation}\label{explicit constituents}
 C_i  = \Span_{\E_i}\bigl\{\bigl(a_{b,0}(\alpha\xi^{u_i}),\ldots, a_{b,\ell-1}(\alpha\xi^{u_i})\bigr): 1\leq b \leq r \bigr\}.   
\end{equation}

Note that each field $\E_i$ is isomorphic to a minimal $\lambda$-constacyclic code of length $m$ over $\F_q$; namely, the $\lambda$-constacyclic code in $\F_q^m$ with the irreducible check polynomial $f_i(x)$. If we denote by $\theta_i$ the generating primitive idempotent (see \cite[Theorem 1]{LG}) for the minimal $\lambda$-constacyclic code $\langle \theta_i \rangle$ in consideration, then the isomorphism is given by the maps
\begin{eqnarray} \label{isoms}\hspace{-10pt}
\begin{array}{ccl} \varphi_i:\langle \theta_i \rangle
& \longrightarrow & \E_i \\ \hspace{0.5cm} a(x)& \longmapsto &
a(\alpha\xi^{u_i}) \end{array}
&\hspace{-16pt} & \begin{array}{ccl} \psi_i: \E_i & \longrightarrow & \langle \theta_i \rangle \\
\hspace{0.5cm} \delta & \longmapsto & \sum\limits_{k=0}^{m-1} a_kx^k,
\end{array}
\end{eqnarray}
where
$$a_k=\frac{1}{m} \Tr_{\E_i/\F_q}(\delta\alpha^{-k}\xi^{-ku_i}).\vspace{8pt}$$
Observe that, for each $i\in\{1,\ldots,s\}$, the maps $\varphi_i$ and $\psi_i$ are inverses of each other, regardless of the choice of the representative in the $\Fq$-conjugacy class $\big[\alpha\xi^{u_i}\big]$, since $\Tr_{\E_i/\F_q}(\epsilon^q)=\Tr_{\E_i/\F_q}(\epsilon)$, for any $\epsilon\in\E_i$.

If $\mathfrak{C}_i$ is a linear code in $\E_i^{\ell}$, for each $i$, then we denote its concatenation with $\langle \theta_i \rangle$ by $\langle \theta_i \rangle \Box \mathfrak{C}_i$ and the concatenation is carried out by the map $\psi_i$. Here, $\langle \theta_i \rangle$ and $\mathfrak{C}_i$ are called the inner and outer codes of the concatenation, respectively.

\begin{thm} \cite[Theorem 2]{LG} \label{Jensen's thm}\hfill
\begin{enumerate}
\item[i.] Let $C$ be an $R$-submodule of $R^\ell$ ({\it i.e.}, a $q$-ary $\lambda$-QT code). Then, for some subset $\mathcal{I}$ of $\{1,\ldots ,s\}$, there exist linear codes $\mathfrak{C}_i\subseteq\E_i^{\ell}$ such that $$C=\displaystyle\bigoplus_{i\in \mathcal{I}} \langle \theta_i \rangle \Box \mathfrak{C}_i.$$
\item[ii.] Conversely, let $\mathfrak{C}_i$ be a linear code over $\E_i$ of length $\ell$, for each $i\in \mathcal{I} \subseteq \{1,\ldots ,s\}$. Then, $$C=\displaystyle\bigoplus_{i\in \mathcal{I}} \langle \theta_i \rangle \Box \mathfrak{C}_i$$ is a $q$-ary $\lambda$-QT code of length $m\ell$ and index $\ell$.
\end{enumerate}
\end{thm}

Moreover, each constituent $C_i$ in (\ref{constituents}) is equal to the outer code $\mathfrak{C}_i$ in the concatenated structure, for each $i$ (see \cite[Theorem 3]{LG}).

By (\ref{isoms}) and Theorem \ref{Jensen's thm}, an arbitrary codeword $\mathbf{c}\in C$ can be written as an $m\times \ell$ array in the form (see \cite{LG})

\begin{equation}\label{QT tr}\setlength\arraycolsep{0.1pt}
\mathbf{c}=\frac{1}{m}\left(\begin{array}{c}
\left(\sum\limits_{i=1}^{s}\Tr_{\E_{i}/\F_q}\left(\kappa_{i,t}\alpha^{-0}\xi^{-0u_{i}} \right) \right)_{0\leq t \leq \ell-1} \\
\left(\sum\limits_{i=1}^{s}\Tr_{\E_{i}/\F_q}\left(\kappa_{i,t}\alpha^{-1}\xi^{-u_{i}} \right) \right)_{0\leq t \leq \ell-1} \\
 \vdots \\
\left(\sum\limits_{i=1}^{s}\Tr_{\E_{i}/\F_q}\left(\kappa_{i,t}\alpha^{-(m-1)}\xi^{-(m-1)u_{i}}
\right) \right)_{0\leq t \leq \ell-1}
\end{array} \right),\\
\end{equation} 
where $\mathbf{\kappa}_i=(\kappa_{i,0},\ldots ,\kappa_{i,\ell-1}) \in C_i$, for all $i$. Since $mC=C$, every codeword in $C$ can still be written in the form of (\ref{QT tr}) with the constant $\frac{1}{m}$ removed. 

Note that the trace representation in (\ref{QT tr}) involves traces down to $\Fq$ from various extensions. Recall that $\F$ is the splitting field of $x^m-\lambda$ and each $\E_i$ is an intermediate field extension between $\F$ and $\Fq$, for $1\leq i\leq s$. The next lemma enables us to rewrite the traces in (\ref{QT tr}) from $\F$ to $\Fq$, instead of treating them over different extensions.

\begin{lem}\cite[Lemma 4.1]{GO}\label{trace lem}
Let $\Fq\subset\K\subset\LL$ be field extensions. If $b\in\LL$ is an element with $\Tr_{\LL/\K}(b)=1$, then we have $\Tr_{\LL/\F_q}(b\mu)=\Tr_{\K/\F_q}(\mu)$, for any $\mu\in\K$.
\end{lem}

Using Lemma \ref{trace lem}, we can rewrite (\ref{QT tr}), without the constant $\frac{1}{m}$, as (cf. (4.3) in \cite{GO})
\begin{equation}\label{QT tr-2}\setlength\arraycolsep{0.1pt}
\mathbf{c}=\left(\begin{array}{c}
\left(\Tr_{\F/\F_q}\left(\sum\limits_{i=1}^{s}b_i\kappa_{i,t}\alpha^{-0}\xi^{-0u_{i}} \right) \right)_{0\leq t \leq \ell-1} \\
\left(\Tr_{\F/\F_q}\left(\sum\limits_{i=1}^{s}b_i\kappa_{i,t}\alpha^{-1}\xi^{-u_{i}} \right) \right)_{0\leq t \leq \ell-1} \\
 \vdots \\
\left(\Tr_{\F/\F_q}\left(\sum\limits_{i=1}^{s}b_i\kappa_{i,t}\alpha^{-(m-1)}\xi^{-(m-1)u_{i}}
\right) \right)_{0\leq t \leq \ell-1}
\end{array} \right),
\end{equation} 
where $b_1,\ldots,b_s\in\F$ are such that $\Tr_{\F/\E_i}(b_i)=1$, for each $i\in\{1,\ldots,s\}$. Such $b_i$'s exist since the trace map is onto.

Jensen derived a minimum distance bound in \cite[Theorem 4]{J}, which is valid for all concatenated codes ({\it i.e.}, the inner and outer codes can be any linear code). Therefore, it applies to QT codes as well. We formulate the Jensen bound for QT codes as follows.
\begin{thm}  \label{Jensen bound}
Let $C\subseteq R^\ell$ be a $\lambda$-QT code with the concatenated structure $C=\bigoplus_{i\in \mathcal{I}} \langle \theta_i \rangle \Box C_i$, for some $\mathcal{I}\subseteq\{1,\ldots ,s\}$. Assume that $C_{i_1}, \ldots, C_{i_t}$ are the nonzero outer codes (constituents) of $C$, for $\{i_1,\ldots,i_t\} \subseteq \mathcal{I}$, such that $d(C_{i_1}) \leq d(C_{i_2})\leq \cdots \leq d(C_{i_t})$. Then we have
\begin{equation}\label{Jensen}
d(C) \geq \displaystyle \min_{1\leq r \leq t} \left\{ d(C_{i_r}) d(\langle\theta_{i_1}\rangle \oplus \cdots \oplus \langle\theta_{i_r}\rangle) \right\}.
\end{equation}
\end{thm}
\vspace{7pt}

\subsection{Spectral theory for quasi-twisted codes}    \label{spec sect}

Lally and Fitzpatrick proved in \cite{LF} that every QC code has a polynomial generating set in the form of a reduced Gr{\"o}bner basis. We provide an easy adaptation of their findings for QT codes. 

Consider the ring homomorphism:
\begin{eqnarray}\label{embedding}
\Psi \ :\ \Fq[x]^{\ell} &\longrightarrow& R^{\ell} \\\nonumber
(\widetilde{f}_0(x), \ldots ,\widetilde{f}_{\ell-1}(x))  &\longmapsto & 
(f_0(x), \ldots ,f_{\ell-1}(x)),
\end{eqnarray}
which projects elements in $\Fq[x]^{\ell}$ onto $R^{\ell}$ in the obvious way. Given a $\lambda$-QT code $C\subseteq R^{\ell}$, it follows that the preimage $\widetilde{C}$ of $C$ in $\Fq[x]^{\ell}$ is an $\Fq[x]$-submodule containing $\widetilde{K} =\{(x^m-\lambda)\mathbf{e}_j : 0\leq j \leq \ell-1\}$, where each $\mathbf{e}_j$ denotes the standard basis vector of length $\ell$ with $1$ at the $j^{\rm th}$ coordinate and $0$ elsewhere. The tilde will represent throughout structures over $\Fq[x]$.

Since $\widetilde{C}$ is a submodule of the finitely generated free module $\F_q[x]^{\ell}$ over the principal ideal domain $\Fq[x]$ and contains $\widetilde{K}$, it has a generating set of the form $$\{\mathbf{u}_1,\ldots,\mathbf{u}_p, (x^m-\lambda)\mathbf{e}_0,\ldots,(x^m-\lambda)\mathbf{e}_{\ell-1}\},$$  where $p$ is a nonnegative integer and when $p>0$, $\mathbf{u}_b = (u_{b,0}(x),\ldots,u_{b,\ell-1}(x))\in \Fq[x]^{\ell}$, for each $b \in \{1,\ldots,p\}$. 
Hence, the rows of

$$\mathcal{G}=\left(\begin{array}{ccc}
    u_{1,0}(x) & \ldots & u_{1,\ell-1}(x) \\
    \vdots &  & \vdots \\
    u_{p,0}(x) & \ldots & u_{p,\ell-1}(x) \\
     x^m-\lambda & \ldots & 0 \\
    \vdots & \ddots & \vdots \\
    0 & \ldots & x^m-\lambda \\
  \end{array}
\right)$$
generate $\widetilde{C}$. By using elementary row operations, we triangularise $\mathcal{G}$ so that another equivalent generating set is obtained from the rows of an upper-triangular $\ell \times \ell$ matrix over $\Fq[x]$ as:
\begin{equation}\label{PGM}
\widetilde{G}(x)=\left(\begin{array}{cccc}
    g_{0,0}(x) & g_{0,1}(x) & \ldots & g_{0,\ell-1}(x) \\
    0 & g_{1,1}(x) & \ldots & g_{1,\ell-1}(x) \\
    \vdots & \vdots & \ddots & \vdots \\
    0 & 0 &\ldots & g_{\ell-1,\ell-1}(x)\\
  \end{array}
\right),
\end{equation}
where $\widetilde{G}(x)$ satisfies the following conditions (see \cite[Theorem 2.1]{LF}):
\begin{enumerate}
    \item $g_{i,j}(x)=0$, for all $0\leq j < i \leq \ell-1$.
    \item deg$(g_{i,j}(x)) < $ deg$(g_{j,j}(x))$, for all $i <j$.
    \item $g_{j,j}(x) \mid (x^m-\lambda)$, for all $0\leq j \leq \ell-1$.
    \item If  $g_{j,j}(x) = (x^m-\lambda)$, then $g_{i,j}(x) =0$, for all $i\neq j$.
\end{enumerate}
Note that the rows of $\widetilde{G}(x)$ are nonzero and each nonzero element of $\widetilde{C}$ can be expressed in the form \begin{center}$(0,\ldots,0,c_j(x),\ldots,c_{\ell-1}(x)),$ where $j\geq 0$, $ c_j(x)\neq 0$ and $g_{j,j}(x)\mid c_j(x)$. \end{center} This implies that the rows of $\widetilde{G}(x)$ form a Gr\"obner basis of $\widetilde{C}$ with respect to the position-over-term (POT) order in $\Fq[x]$, where the standard basis vectors $\{\mathbf{e}_0,\ldots,\mathbf{e}_{\ell-1}\}$ and the monomials $x^n$ are ordered naturally in each component. Moreover, the second condition above implies that the rows of $\widetilde{G}(x)$ provide a reduced Gr\"obner basis for $\widetilde{C}$, which is uniquely defined, up to multiplication by constants, with monic diagonal elements.

Let $G(x)$ be the matrix with the rows of $\widetilde{G}(x)$ under the image of the homomorphism $\Psi$ in (\ref{embedding}). Clearly, the rows of $G(x):=\widetilde{G}(x) \mod I$ is an $R$-generating set for $C$. When $C$ is the trivial zero code of length $m\ell$, we have $p=0$, which gives $G(x)=\mathbf{0}_{\ell}$. Otherwise, we say that $C$ is an $r$-generator $\lambda$-QT code, generated as an $R$-submodule, if $G(x)$ has $r$ (nonzero) rows. The $\Fq$-dimension of $C$ is given by (see \cite[Corollary 2.4]{LF} for the proof) 
\begin{equation}\label{dimension}
m\ell-\sum_{j=0}^{\ell-1}\mbox{deg}(g_{j,j}(x))=\sum_{j=0}^{\ell-1}\left[m -\mbox{deg}(g_{j,j}(x))\right].
\end{equation}

In \cite{ST}, Semenov and Trifonov used the polynomial matrix $\widetilde{G}(x)$ in (\ref{PGM}) to develop a spectral theory for QC codes, which gives rise to a BCH-like minimum distance bound. Their bound was improved by Zeh and Ling in \cite{LZ} by using the HT bound (\cite{HT}), which generalizes the BCH bound for cyclic codes. We now translate their results from QC to QT codes. 

Given a $\lambda$-QT code $C\subseteq R^{\ell}$, let the associated $\ell \times \ell$ upper-triangular matrix $\widetilde{G}(x)$ be as in (\ref{PGM}) with entries in $\Fq[x]$. The {\it determinant} of $\widetilde{G}(x)$ is defined as $$\mbox{det}(\widetilde{G}(x)):=\prod_{j=0}^{\ell-1}g_{j,j}(x)$$ and an {\it eigenvalue} $\beta$ of $C$ is a root of det$(\widetilde{G}(x))$. Note that all eigenvalues are elements of $\Omega$ ({\it i.e.}, $\beta=\alpha\xi^k$, for some $k\in\{0,\ldots,m-1\}$), since $g_{j,j}(x)\mid (x^m-\lambda)$, for each $0\leq j\leq \ell-1$. The {\it algebraic multiplicity} of $\beta$ is the largest integer $a$ such that $(x-\beta)^a\mid\mbox{det}(\widetilde{G}(x))$. The {\it geometric multiplicity} of $\beta$ is defined as the dimension of the null space of $\widetilde{G}(\beta)$, where this null space is called the {\it eigenspace} of $\beta$ and it is denoted by $\mathcal{V}_{\beta}$. In other words, we have $$\mathcal{V}_{\beta}:=\{\mathbf{v}\in\F^{\ell} : \widetilde{G}(\beta)\mathbf{v}^{\top}=\mathbf{0}_{\ell}^{\top}\},$$ where $\F$ is the splitting field of $x^m-\lambda$ as before. Semenov and Trifonov showed in \cite{ST} that, for a given QC code and the associated $\widetilde{G}(x) \in \F_q[x]^{\ell\times\ell}$, the algebraic multiplicity $a$ of an eigenvalue $\beta$ is equal to its geometric multiplicity $\mbox{dim}_{\F}(\mathcal{V}_{\beta})$. We state the QT analogue of this result below, which can be shown in the same way and therefore the proof is omitted.
\begin{lem}\cite[Lemma 1]{ST}\label{multiplicity lemma}
The algebraic multiplicity of any eigenvalue of a $\lambda$-QT code $C$ is equal to its geometric multiplicity.
\end{lem}

Throughout, we let $\overline{\Omega}\subseteq \Omega$ denote the set of all eigenvalues of $C$. Notice that $\overline{\Omega}=\emptyset$ if and only if the diagonal elements $g_{j,j}(x)$ in $\widetilde{G}(x)$ are constant polynomials and $C$ is the trivial full space code. From this point on, we exclude the full space code and we assume that $|\overline{\Omega}|=t>0$. Choose an arbitrary eigenvalue $\beta_i\in\overline{\Omega}$ with multiplicity $n_i$, for some $i \in\{1,\ldots,t\}$. Let $\{\mathbf{v}_{i,0},\ldots,\mathbf{v}_{i,n_i-1}\}$ be a basis for the corresponding eigenspace $\mathcal{V}_i$. Consider the matrix  
\begin{equation}\label{Eigenspace} 
V_i:=\begin{pmatrix}
\mathbf{v}_{i,0} \\
\vdots\\
\mathbf{v}_{i,n_i-1} 
 \end{pmatrix}=\begin{pmatrix}
v_{i,0,0}&\ldots&v_{i,0,\ell-1} \\
\vdots &\vdots & \vdots\\
v_{i,n_i-1,0}&\ldots&v_{i,n_i-1,\ell-1}
 \end{pmatrix},
\end{equation} 
having the basis elements as its rows. We let
\[ H_i:=(1, \beta_i,\ldots,\beta_i^{m-1})\otimes V_i\ \]
and define
\begin{equation}\label{parity check matrix} 
H:=\begin{pmatrix}
H_1 \\
\vdots\\
H_t 
\end{pmatrix}=\begin{pmatrix}
V_1&\beta_1V_1&\ldots&\beta_1^{m-1}V_1 \\
\vdots &\vdots & & \vdots\\
V_t&\beta_tV_t&\ldots&\beta_t^{m-1}V_t
\end{pmatrix}.
\end{equation} 
Observe that $H$ has $n:=\sum_{i=1}^t n_i$ rows. By Lemma \ref{multiplicity lemma}, we have $n=\sum_{j=0}^{\ell-1}\mbox{deg}(g_{j,j}(x))$. To prove the following lemma, it remains to show that all these $n$ rows are linearly independent, which was already shown in \cite[Lemma 2]{ST}.
\begin{lem}\label{rank lemma}
The rank of the matrix $H$ in (\ref{parity check matrix}) is equal to $m\ell -\dim_{\Fq}(C)$.
\end{lem}

We observe that $H \mathbf{c}^{\top}=\mathbf{0}_n^{\top}$, for any codeword $\mathbf{c}\in C$. Together with Lemma \ref{rank lemma}, we easily obtain the following result (see \cite[Theorem 1]{ST} for the QC analogue of the result).

\begin{prop}
The $n\times m\ell$ matrix $H$ in (\ref{parity check matrix}) is a parity-check matrix for $C$.
\end{prop}

\begin{rem}\label{analogy rem}
The eigenvalues are the QT analogues of the zeros of constacyclic codes. Recall that a constacyclic code has an empty zero set if and only if it is equal to the full space. Similarly, $\overline{\Omega}=\emptyset$ if and only if $C=\Fq^{m\ell}$. In this case, the construction of the parity-check matrix $H$ in (\ref{parity check matrix}) is impossible, hence, we have assumed $\overline{\Omega}\neq\emptyset$. The other extreme case is when the zero set of a given constacyclic code is $\Omega$, which implies that we have the trivial zero code. However, we emphasize that a $\lambda$-QT code with $\overline{\Omega}=\Omega$ is not necessarily the zero code. By using Lemma \ref{rank lemma} above, one can easily deduce that a given $\lambda$-QT code $C$ is the zero code $\{\mathbf{0}_{m\ell}\}$ if and only if $\overline{\Omega}=\Omega$, each $\mathcal{V}_i=\F^{\ell}$ (equivalently, each $V_i=I_{\ell}$, where $I_{\ell}$ denotes the $\ell\times\ell$ identity matrix), and $n=m\ell$ so that we obtain $H=I_{m\ell}$. On the other hand, $\overline{\Omega}=\Omega$ whenever $(x^m-\lambda) \mid \mbox{det}(\widetilde{G}(x))$ but $C$ is nontrivial unless each $m^{\rm th}$ root of $\lambda$ in $\Omega$ has multiplicity $\ell$, which happens only if $\widetilde{G}(x)=(x^m-\lambda) I_{\ell}$.
\end{rem}

\begin{defn}\label{eigencode}
We define the {\it eigencode} corresponding to an eigenspace $\mathcal{V}\subseteq \F^\ell$ by
\[
\mathbb{C}(\mathcal{V})=\mathbb{C}:=\left\{\mathbf{u}\in \Fq^\ell\ : \ \sum_{j=0}^{\ell-1}{v_ju_j}=0, \forall \mathbf{v} \in \mathcal{V}\right\}.
\]
In case we have $\mathbb{C}=\{\mathbf{0}_{\ell}\}$, then it is assumed that $d(\mathbb{C})=\infty$.
\end{defn}

Semenov and Trifonov proved a BCH-like minimum distance bound for a given QC code (see \cite[Theorem 2]{ST}), which is expressed in terms of the size of a consecutive subset of eigenvalues in $\overline{\Omega}$ and the minimum distance of the common eigencode related to this consecutive subset. Zeh and Ling generalized their approach and derived an HT-like bound in \cite[Theorem 1]{LZ} without using the parity-check matrix in their proof. The eigencode, however, is still needed. In \cite{ELOT}, a general spectral bound is proven for any QT code with a nonempty set of eigenvalues that is different from $\Omega$. The QT analogues of Semenov-Trifonov and Zeh-Ling bounds were proven in terms of the Roos-like and shift-like bounds as a corollary. The observations in the next section are crucial for extending this bound to the general case of any nonempty set of eigenvalues.

\section{Eigencodes and Constituents} \label{comparison section}

Recall the factorization into irreducibles $x^m-\lambda=f_1(x)\cdots f_s(x)$ given in (\ref{irreducibles}). If we fix a root $\alpha\xi^{u_i}$ of $f_i(x)$, for each $i\in\{1,\ldots,s\}$, then we know that $\Omega$ is the disjoint union of the $\Fq$-conjugacy classes $\big[\alpha\xi^{u_i}\big]$, each of the form as in (\ref{cyc coset}) and of size $e_i=$ deg($f_i$). By Theorem \ref{Jensen's thm}, any $\lambda$-QT code $C$, viewed as an $R$-submodule of $R^\ell$, decomposes as 
\begin{equation}\label{decomp}
C=\displaystyle\bigoplus_{i=1}^s \langle \theta_i\rangle \Box C_i,
\end{equation}
where each inner code $\left\langle \theta_i \right\rangle $ is a minimal $\lambda$-constacyclic code of length $m$ satisfying $\left\langle \theta_i \right\rangle \cong \E_i \cong \Fq[x]/\langle f_i(x)\rangle$ (see (\ref{isoms})) and each outer code (constituent) $C_i$ is a linear code in $\E_i^{\ell}$, for $i\in\{1,\ldots,s\}$. By using (\ref{CRT-3}) and (\ref{explicit constituents}), we can write each constituent $C_i$ explicitly as
\begin{equation}\label{explicit}
C_i  =  \bigl\{\mathbf{c}(\alpha\xi^{u_i}) : \mathbf{c}(x)\in C\}.   
\end{equation}

On the other hand, any codeword $\mathbf{c}(x)\in C$ is of the form $\mathbf{c}(x)=\mathbf{a}(x)\widetilde{G}(x) \mod I$, for some $\mathbf{a}(x)\in \Fq[x]^{\ell}$. Hence, the $\E_i$-span of the rows of $\widetilde{G}(\alpha\xi^{u_i})$ is the constituent $C_i$, for each $i$. If all the diagonal elements in $\widetilde{G}(\alpha\xi^{u_i})$ are nonzero ({\it i.e.}, $\widetilde{G}(\alpha\xi^{u_i})$ has full rank $\ell$), then $C_i=\E_i^{\ell}$ and $\alpha\xi^{u_i}$ is \emph{not} an eigenvalue of $C$. Otherwise, $\alpha\xi^{u_i}$ is an eigenvalue of $C$ and the corresponding nonzero eigenspace $\mathcal{V}_i$ is described as $\mathcal{V}_i=\{\mathbf{v}\in\F^{\ell} : \widetilde{G}(\alpha\xi^{u_i}) \mathbf{v}^{\top}=\mathbf{0}_{\ell}^{\top}\}$. Now let $\overline{C}_i$ denote the $\F$-span of the rows of $\widetilde{G}(\alpha\xi^{u_i})$, for each $i$, which immediately implies $\mathcal{V}_i=\overline{C}_i^{\perp}$. Clearly, each constituent $C_i$ is the $\E_i$-subcode of $\overline{C}_i$, for all $i$. Note that the $\F$-span of the rows of $\widetilde{G}(\alpha^q\xi^{qu_i})$ is
\begin{eqnarray*}
\overline{C}_i^q &=& \{(d_0^q,\ldots,d_{\ell-1}^q) : (d_0,\ldots,d_{\ell-1}) \in \overline{C}_i\}\\ &\supseteq& C_i^q=\{\mathbf{c}(\alpha^q\xi^{qu_i}) : \mathbf{c}(x)\in C\},
\end{eqnarray*}
for any $i\in\{1,\ldots,s\}$. Since each entry of $\widetilde{G}(\alpha^q\xi^{qu_i})$ is the $q^{\rm th}$ power of the corresponding entry in $\widetilde{G}(\alpha\xi^{u_i})$, for each $i$, the dual of $\overline{C}_i^q$ consists of elements $\mathbf{v}^q=(v_0^q,\ldots,v_{\ell-1}^q)$, where $\mathbf{v}\in\mathcal{V}_i$,  which implies the following.
\begin{lem} \label{eigenspace lemma}
If $\alpha\xi^{u_i}$ is an eigenvalue of $C$, for some $i\in\{1,\ldots,s\}$, then all elements in its $\Fq$-conjugacy class $\big[\alpha\xi^{u_i}\big]$ are also eigenvalues with respective eigenspaces $[\mathcal{V}_i]:=\left\{\mathcal{V}_i,\mathcal{V}_i^q,\ldots,\mathcal{V}_i^{q^{e_i-1}}\right\}$. 
\end{lem}
Recall that, for each $1\leq i\leq s$, the eigencode $\mathbb{C}_i$ corresponding to the eigenspace $\mathcal{V}_i$ is given as (cf. Definition \ref{eigencode}) $$\mathbb{C}_i=\left\{\mathbf{u}\in \Fq^\ell\ : \ \sum_{j=0}^{\ell-1}{v_ju_j}=0, \forall \mathbf{v} \in \mathcal{V}_i\right\}.$$ The fact that $\mathcal{V}_i=\overline{C}_i^{\perp}$ implies $\mathbb{C}_i=\overline{C}_i \big\vert_{\Fq}$. Obviously, $\overline{C}_i \big\vert_{\Fq}=\overline{C}_i^q \big\vert_{\Fq}=\cdots=\overline{C}_i^{q^{e_i-1}} \big\vert_{\Fq}$ and, therefore, $\mathbb{C}_i$ is the common eigencode of all the eigenspaces in $[\mathcal{V}_i]$.
\begin{lem} \label{eigencode lemma}
For each $1\leq i\leq s$, the eigencode $\mathbb{C}_i$ corresponding to the eigenspaces in $[\mathcal{V}_i]$ is the $\Fq$-subcode of $\overline{C}_i$.
\end{lem}

Now we reorder the constituents. Let $C_1=\cdots=C_r=\{ \mathbf{0}_{\ell} \}$ be the zero constituents. Let  $C_{r+1}=\cdots=C_t$ be the full space constituents, where $\widetilde{G}(\alpha\xi^{u_i})$ has full rank $\ell$, for $r+1\leq i \leq t$. Note that $\widetilde{G}(\alpha\xi^{u_i})$ having full rank $\ell$ is equivalent to $\alpha\xi^{u_i}$ not being an eigenvalue and, therefore, $\big[\alpha\xi^{u_i}\big]\notin \overline{\Omega}$, for $r+1\leq i \leq t$. Finally, let $C_{t+1},\ldots,C_s$ be the nontrivial constituents, {\it i.e.}, they are neither the full space codes nor the zero codes. We have the disjoint unions 
\begin{equation}\label{sets}
\Omega =\bigcup_{i=1}^s\big[\alpha\xi^{u_i}\big] \mbox{\quad and\quad} \Omega \setminus\overline{\Omega} =\bigcup_{i=r+1}^t\big[\alpha\xi^{u_i}\big].
\end{equation}

Let $\Gamma$ denote the set of indices 
\begin{equation} \label{notation gamma}
\Gamma:=\{1,\ldots,r,t+1,\ldots,s\}=\{1,\ldots, s\}\setminus \{r+1,\ldots,t\}
\end{equation}
such that $\overline{\Omega} = \bigcup_{i\in\Gamma}\big[\alpha\xi^{u_i}\big]$, by (\ref{sets}). The common eigenspace $\mathcal{V}_{\overline{\Omega}}$ satisfies 
$$\mathcal{V}_{\overline{\Omega}}=\bigcap_{\beta\in\overline{\Omega}} \mathcal{V}_{\beta}= \bigcap_{i\in\Gamma}\bigcap_{j=0}^{e_i-1} \mathcal{V}_i^{q^j}$$
and the associated eigencode is 
\begin{eqnarray} \label{common eigencode}
\mathbb{C}_{\overline{\Omega}}=(\mathcal{V}_{\overline{\Omega}})^{\perp} \biggr\rvert_{\Fq}&=&\left(\sum_{i\in\Gamma} \sum_{j=0}^{e_i-1}(\mathcal{V}_i^{q^j})^{\perp}\right) \biggr\rvert_{\Fq}\\ \notag
&=&\left(\sum_{i\in\Gamma} \sum_{j=0}^{e_i-1} \overline{C}_i^{q^j}\right) \biggr\rvert_{\Fq}.
\end{eqnarray} 

A linear code $\mathcal{C}$ of length $n$ over a field $\K \supseteq \Fq$ is called \emph{Galois closed} if $\mathcal{C}=\mathcal{C}^q$, where $$\mathcal{C}^q := \left\{(c_0^q,c_1^q,\ldots,c_{n-1}^q)\, :\, (c_0,c_1,\ldots,c_{n-1})\in \mathcal{C}\right\}.$$ With this definition, it is easy to observe that $(\mathcal{V}_{\overline{\Omega}})^{\perp}$ is Galois closed. Theorem 12.17 in \cite{B} says $d(\mathcal{C})=d(\mathcal{C}|_{\Fq})$ if $\mathcal{C}$ is Galois closed. Therefore, we conclude that $d((\mathcal{V}_{\overline{\Omega}})^{\perp})=d(\mathbb{C}_{\overline{\Omega}})$. Note that each constituent code $C_i$, for $i\in\Gamma$, is the subfield subcode of a summand $\overline{C}_i$ in $(\mathcal{V}_{\overline{\Omega}})^{\perp}$ and it is a well-known fact that $d(A)\geq d(A+B)$, for any pair of linear codes $A$ and $B$. Hence, we obtain the following.

\begin{lem}\label{cons-eig lemma}
Let $C$ be a $\lambda$-QT code with a nonempty eigenvalue set $\overline{\Omega}\subseteq \Omega$ which decomposes as in (\ref{decomp}). Then \begin{center}$d(C_i)\geq d(\overline{C}_i) \geq  d\left(\left(\mathcal{V}_{\overline{\Omega}}\right)^{\perp}\right) = d\left(\left(\mathcal{V}_{\overline{\Omega}}\right)^{\perp} \big\rvert_{\Fq}\right) =  d(\mathbb{C}_{\overline{\Omega}})$,\end{center} 
for each $i \in \Gamma$.
\end{lem}

Equation (\ref{common eigencode}) will play a key role in the next two sections. It is part of the proof of the general spectral bound in Section \ref{bound sect}. Together with Lemma \ref{cons-eig lemma}, the equation allows us to compare the Jensen and spectral bounds in Section \ref{res sect}, where an analogous relation between the Lally and spectral bounds follows as a result of similar observations.

\section{Spectral bound for quasi-twisted codes}\label{bound sect}

The general spectral bound proven in \cite[Theorem 11]{ELOT} was shown for QT codes with nonempty eigenvalue set $\overline{\Omega}\varsubsetneqq\Omega$. In fact, the bound remains valid when $\overline{\Omega}=\Omega$. To show this we need the following Lemma.

\begin{lem}\label{freaking lemma}
Let $C$ be a $\lambda$-QT code of length $m\ell$ and index $\ell$ over $\F_q$ with $\overline{\Omega}=\Omega$. Then $C=\{\mathbf{0}_{m\ell}\}$  if and only if $\mathbb{C}_{\Omega}=\{\mathbf{0}_{\ell}\}$.
\end{lem}

\begin{proof}
Let $B$ be the $q$-ary linear code of length $\ell$ that is generated by the rows of the codewords in $C$, which are represented as the $m\times \ell$ arrays in (\ref{array}) or, equivalently, as in (\ref{QT tr-2}). That is, the codewords of $B$ are generated by the elements $$\left(\Tr_{\F/\F_q}\left(\sum\limits_{i=1}^{s}b_i\kappa_{i,t} \, \alpha^{-k} \, \xi^{-ku_{i}}\right) \right)_{0\leq t \leq \ell-1},$$
where $k\in \{0,\ldots,m-1\}$, $b_i\in\F$ such that $\Tr_{\F/\E_i}(b_i)=1$, and each $\kappa_i=(\kappa_{i,0},\ldots ,\kappa_{i,\ell-1}) $ is a codeword in the constituent $C_i$, for all $i$.

Since we have assumed that the $\lambda$-QT code $C$ has the eigenvalue set $\overline{\Omega}=\Omega$, we cannot have full space constituents. In order to have a constituent code that has full rank, there must be at least one $m^{\rm th}$ root of $\lambda$ which is not a root of any diagonal element in the upper-triangular matrix $\widetilde{G}(x)$ that corresponds to $C$. Together with the notation provided in (\ref{notation gamma}), we rewrite the generators of $B$ as 
\begin{equation}\label{subtrace}
\left(\Tr_{\F/\F_q}\left(\sum\limits_{i\in\Gamma}b_i\kappa_{i,t} \, \alpha^{-k} \, \xi^{-ku_{i}}\right) \right)_{0\leq t \leq \ell-1}.
\end{equation}

On the other hand, recall that $\mathbb{C}_{\overline{\Omega}}=(\mathcal{V}_{\overline{\Omega}})^{\perp} |_{\Fq}$, where $(\mathcal{V}_{\overline{\Omega}})^{\perp}$ is Galois closed. Theorem 12.17 in \cite{B} tells us that, if a linear code $\mathcal{C}$ over a field $\K \supseteq \Fq$ is Galois closed, then $\mathcal{C}|_{\Fq}=\Tr_{\K/\F_q}(\mathcal{C})$. Using this property, we rewrite (\ref{common eigencode}) as
\begin{align}\label{common eigencode-2}
\mathbb{C}_{\Omega} &= \mathbb{C}_{\overline{\Omega}}=(\mathcal{V}_{\overline{\Omega}})^{\perp} \biggr\rvert_{\Fq} 
= \left(\sum_{i\in\Gamma} \sum_{j=0}^{e_i-1}(\mathcal{V}_i^{q^j})^{\perp}\right) \biggr\rvert_{\Fq} \notag \\
&=\Tr_{\F/\Fq}\left(\sum_{i\in\Gamma} \sum_{j=0}^{e_i-1} (\mathcal{V}_i^{q^j})^{\perp}\right) \notag \\
&=\Tr_{\F/\Fq}\left(\sum_{i\in\Gamma} \sum_{j=0}^{e_i-1} \overline{C}_i^{q^j}\right).
\end{align}

It is clear that if $C=\{\mathbf{0}_{m\ell}\}$, then $\mathcal{V}_{\Omega}=\F^{\ell}$ and therefore $\mathbb{C}_{\Omega}=\{\mathbf{0}_{\ell}\}$. For the converse, note that, for each $i\in\Gamma$, we have $b_i\kappa_i\in\overline{C}_i$ since $\kappa_i\in C_i$ and $b_i\in\F$. Using (\ref{subtrace}) and (\ref{common eigencode-2}), one can easily see that $B$ is a subcode of $\mathbb{C}_{\Omega}$. Hence, if $\mathbb{C}_{\Omega}=\{\mathbf{0}_{\ell}\}$, then $B=\{\mathbf{0}_{\ell}\}$, which immediately implies $C=\{\mathbf{0}_{m\ell}\}$.
\end{proof}

Now we are ready to state and prove a general spectral bound for a given $\lambda$-QT code.
 
\begin{thm}\label{main thm new}
Let $C\subseteq R^\ell$ be a $\lambda$-QT code of index $\ell$ with nonempty eigenvalue set $\overline{\Omega}\subseteq\Omega$, let $D_{\overline{\Omega}}$ be the $\lambda$-constacyclic code of length $m$ over $\F$ with zero set $\overline{\Omega}$ and let $\mathcal{B}(D_{\overline{\Omega}}) \subseteq \mathcal{P}(\Omega) \times (\N \cup \{\infty\})$ be an arbitrary family of defining set bounds for $D_{\overline{\Omega}}$. For any nonempty $P\subseteq \overline{\Omega}$ such that $(P,d_P)\in \mathcal{B}(D_{\overline{\Omega}})$, we define $\mathcal{V}_{P}:=\bigcap_{\beta\in P}\mathcal{V}_{\beta}$ as the common eigenspace of the eigenvalues in $P$ and let $\mathbb{C}_{P}=(\mathcal{V}_{P})^{\perp} \bigr\rvert_{\Fq}$ denote the corresponding eigencode. Then, 
\[
d(C) \geq \min \left\{ d_P, d(\mathbb{C}_P) \right\}.
\]
\end{thm}

\begin{proof}
Given the $\lambda$-QT code $C$ with nonempty eigenvalue set $\overline{\Omega}\subseteq\Omega$, we first consider the case when $P=\overline{\Omega}=\Omega$ and $d_{\Omega}=d(D_{\Omega})$, where the $\lambda$-constacyclic code $D_{\Omega}\subseteq \F^m$ has $\Omega$ as its zero set. Recall that we always have $d_{\Omega}=\infty$ in this case since $D_{\Omega}=\{\mathbf{0}_{m}\}$. Moreover, by Lemma \ref{freaking lemma},  $C$ is the zero code if and only if $\mathbb{C}_{\Omega}=\{\mathbf{0}_{\ell}\}$ if and only if $d(\mathbb{C}_{\Omega})=\infty$. Hence, we have shown that both $d_{\Omega}$ and $d(\mathbb{C}_{\Omega})$ become $\infty$ if and only if $C=\{\mathbf{0}_{m\ell}\}$. If $C$ is not the zero code, then the rows of any codeword $\mathbf{c}\in C$, represented as in (\ref{array}), lie in $\mathbb{C}_{\Omega}$, as shown in the proof of Lemma \ref{freaking lemma}. This immediately implies that $\mbox{wt} (\mathbf{c})\geq d(\mathbb{C}_{\Omega})$, for any nonzero $\mathbf{c}\in C$, and we get $d(C) \geq \min \left\{ d_{\Omega}, d(\mathbb{C}_{\Omega}) \right\}=d(\mathbb{C}_{\Omega})$.

Thus, it remains to prove that $d(C) \geq \min \left\{d_P, d(\mathbb{C}_P) \right\}$, for any fixed $\emptyset \neq P\subseteq \overline{\Omega} \subseteq \Omega$ such that $(P, d_P)\in \mathcal{B}(D_{\overline{\Omega}})$ and $d_P$ is finite. To do this, we assume that $P =\{\alpha\xi^{u_1}, \alpha\xi^{u_2},\ldots,\alpha\xi^{u_r}\}\subseteq \overline{\Omega}$, where $0<r< m$.  We define 
\begin{equation}\label{pmatrix}
\widetilde{H}_P:=\begin{pmatrix}
1&\alpha\xi^{u_1}&(\alpha\xi^{u_1})^2&\ldots&(\alpha\xi^{u_1})^{m-1}\\
\vdots & \vdots & \vdots & \vdots & \vdots \\
1&\alpha\xi^{u_r}&(\alpha\xi^{u_r})^2&\ldots&(\alpha\xi^{u_r})^{m-1}
\end{pmatrix}.
\end{equation}
Recall that $P$ is the zero set of some $D_P\subseteq\F^m$, which contains $D_{\overline{\Omega}}$ as a subcode, and $\widetilde{H}_P$ is a parity-check matrix of this $D_P$. Note that $d(D_{\overline{\Omega}})\geq d(D_P) \geq d_P$, by definition.

In the rest of the proof we focus on the quantity $\min\{d_P, d(\mathbb{C}_P)\}$. We have assumed $P \neq \emptyset$ to ensure that $\widetilde{H}_P$ is well-defined. For any nonzero $\lambda$-QT code $C$, we have $d(C) \ge 1$. In particular, when $\mathcal{V}_P=\{\mathbf{0}_{\ell}\}$, which implies $\mathbb{C}_P= \Fq^\ell$ and $d(\mathbb{C}_P) =1$, we have $d(C) \ge 1 = \min\{d_P,d(\mathbb{C}_P)\}$ since $d_P\geq 1$.

Now assume that $\mathcal{V}_P\neq\{\mathbf{0}_{\ell}\}$ and let $V_P$ be the matrix, say of size $t\times\ell$, whose rows form a basis for the common eigenspace $\mathcal{V}_P$ (cf. (\ref{Eigenspace})). If we set $\widehat{H}_P :=\widetilde{H}_P \otimes V_P$, then $\widehat{H}_P \, \mathbf{c}^{\top}=\mathbf{0}_{rt}^{\top}$, for all $\mathbf{c}\in C$. In other words, $\widehat{H}_P$ is a submatrix of some matrix $H$ of the form in (\ref{parity check matrix}). 

Recall that $d_P$ is assumed to be finite. If $d_P = 1$, then $\min\{d_P, d(\mathbb{C}_P)\}=1$, so $d(C) \ge 1 = \min\{d_P,d(\mathbb{C}_P)\}$ as above. 

When $d_P \ge 2$, we have $\min\{d_P,d(\mathbb{C}_P)\} =1$ if and only if $d(\mathbb{C}_P)=1$, in which case we have $d(C) \ge \min\{d_P,d(\mathbb{C}_P)\}=1$ automatically.

We now let $d_P \ge 2$ and $d(\mathbb{C}_P) \geq 2$ (hence, $\widehat{H}_P$ is well-defined). We assume the existence of a codeword $\mathbf{c}\in C$ of weight $\omega$ such that $0 < \omega < \min\{d_P, d(\mathbb{C}_P)\}$. For each $0 \leq k \leq m-1$, let $\mathbf{c}_k = (c_{k,0}, \ldots, c_{k,\ell-1})$ be the $k^{\rm th}$ row of the codeword $\mathbf{c}$ given as in (\ref{array}) and we consider the column vector $\mathbf{s}_k := V_P \, \mathbf{c}_k^{\top}$. Since $d(\mathbb{C}_P) > \omega$ and $\mbox{wt}(\mathbf{c}_k) \leq \omega$, we have $\mathbf{c}_k \notin\mathbb{C}_P$ and therefore $ \mathbf{s}_k=V_P \, \mathbf{c}_k^{\top}\neq \mathbf{0}_t^{\top}$, for all $\mathbf{c}_k \neq \mathbf{0}_{\ell}$, $k\in\{0,\ldots,m-1\}$. Hence, $0 < \lvert\{\mathbf{s}_k : \mathbf{s}_k \neq \mathbf{0}_t^{\top}\} \rvert \leq \omega < \min\{d_P, d(\mathbb{C}_P)\}$. Let $S := [\mathbf{s}_0\ \mathbf{s}_1 \cdots \mathbf{s}_{m-1}]$. Then $\widetilde{H}_P \, S^{\top} = \mathbf{0}_{r\times t}^{\top}$, which implies that the rows of the matrix $S$ lie in the right kernel of $\widetilde{H}_P$. In other words, any row of $S$ lies in $D_P\subseteq\F^m$ and there is at least one nonzero row in $S$, so by definition of $d_P$, the weight of this nonzero row of $S$ must be at least $d_P$. But this is a contradiction since any row of $S$ has weight at most $\omega < d_P$.
\end{proof}

We emphasize that Theorem \ref{main thm new} allows us to use \emph{any} defining set bound derived for constacyclic codes. The following special cases are immediate after the preparation in Section \ref{basics} (cf. Theorems \ref{Roos} and \ref{shift bound}, and Remark \ref{Roos-shift rem}). The proof is omitted since it is identical to that of \cite[Corollary 12]{ELOT}.

\begin{cor}\label{Cor-Roos-Shift new}
Let $C\subseteq R^\ell$ be a $\lambda$-QT code of index $\ell$ with $\overline{\Omega}\subseteq\Omega$ as its nonempty set of eigenvalues.
\begin{enumerate}
\item[i.] Let $N$ and $M$ be two nonempty subsets of $\Omega$ such that $MN\subseteq\overline{\Omega}$, where $MN:=\frac{1}{\alpha}\bigcup_{\varepsilon\in M} \varepsilon N$. If there exists a consecutive set $M'\supseteq M$ with $|M'|\leq |M|+d_N-2$, then $d(C)\geq \min(|M|+d_N-1,d(\mathbb{C}_{MN}))$.\vspace{3pt}

\item[ii.] For every $A\subseteq\Omega$ that is independent with respect to $\overline{\Omega}$, we have $d(C)\geq \min(|A|,d(\mathbb{C}_{T_A}))$, where $T_A:=A\cap \overline{\Omega}$.
\end{enumerate}
\end{cor}

\begin{rem}\label{ext remark}
By using Remark \ref{Roos remark} and Corollary \ref{shift remark}, we can obtain the QT analogues of the BCH-like bound given in \cite[Theorem 2]{ST} and the HT-like bound in \cite[Theorem 1]{LZ}. 
\end{rem}

Let the $\lambda$-QT code $C$ with nonempty eigenvalue set $\overline{\Omega}\subseteq\Omega$ and the associated $\lambda$-constacyclic code $D_{\overline{\Omega}}\subseteq \F^m$ with the selected collection of defining set bounds $\mathcal{B}(D_{\overline{\Omega}})$ be given as in Theorem \ref{main thm new}. From this point on, we denote the estimate of the spectral bound by $$d_{Spec}(\mathcal{B}(D_{\overline{\Omega}}) ; P,d_P):=\min\{d_P, d(\mathbb{C}_P)\},$$ where $\emptyset\neq P \subseteq \overline{\Omega}$ such that $(P,d_P)\in\mathcal{B}(D_{\overline{\Omega}})$, and we set 
\[
d_{Spec}(\mathcal{B}(D_{\overline{\Omega}})):=\displaystyle\max_{\substack{(P, d_P)\in \mathcal{B}(D_{\overline{\Omega}}) \\ \emptyset\neq P \subseteq \overline{\Omega}}}\left\{d_{Spec}(\mathcal{B}(D_{\overline{\Omega}}) ; P,d_P)\right\}.
\]
Observe that $d_{Spec}(\mathcal{B}(D_{\overline{\Omega}}))$ and $d_{Spec}(\mathcal{B}(D_{\overline{\Omega}}) ; P,d_P)$, for any $\emptyset\neq P \subseteq \overline{\Omega}\subseteq\Omega$ with $(P, d_P)\in \mathcal{B}(D_{\overline{\Omega}})$, are well-defined for any nontrivial $\lambda$-QT code $C$. Namely, $d_{Spec}(\mathcal{B}(D_{\overline{\Omega}}))=\infty$ if and only if $C$ is the zero code, which is the only case when $\overline{\Omega}=\Omega$ and $d_{Spec}(\mathcal{B}(D_{\Omega}) ; \Omega,d_{\Omega})=\infty$, by Lemma \ref{freaking lemma}. 

\section{Comparison results} \label{res sect}

In \cite[Section IV]{ST}, Semenov and Trifonov gave a comparison of their BCH-like spectral bound with the Lally bound in \cite{L2}, the Barbier-Chabot-Quintin bound in \cite{Barbier}, and the Tanner bound in \cite{Tanner}, all for QC codes. Here, we consider the performance of the generalized spectral bound given in Theorem \ref{main thm new} against the Jensen bound in Theorem \ref{Jensen bound}. A similar performance comparison against the Lally bound for QT codes will be provided right after.

\subsection{Jensen against Spectral} \label{JS sect}

Given $C$ with the concatenated structure as in (\ref{decomp}), recall the ordering of the constituents in Section \ref{comparison section}, where $C_1=\cdots=C_r=\{ \mathbf{0}_{\ell} \}$ are the zero constituents, $C_{r+1}=\cdots=C_t$ are the full space constituents, and $C_{t+1},\ldots,C_s$ are the nontrivial constituents. Without loss of generality, we assume that $1 \leq d(C_{t+1}) \leq \cdots \leq d(C_s)$.

With this grouping of the constituents, we obtain the following ordering of their distances 
$$1 = d(C_{r+1})=\cdots=d(C_t) \leq d(C_{t+1})\leq\cdots\leq d(C_s),$$
which allows us to rewrite the Jensen bound (given in (\ref{Jensen})) in an organized way as
\begin{equation}\label{Jensen2}
d(C) \geq \displaystyle\min_{r+1\leq i \leq s} \left\{ d(C_{i}) d(\langle\theta_{r+1}\rangle \oplus \cdots \oplus \langle\theta_{i}\rangle) \right\}.
\end{equation}

In other words, the Jensen bound is the minimum of the following distances
$$\setlength\arraycolsep{1pt}
\begin{array}{rl}
d(C_{r+1}) & d(\langle \theta_{r+1}\rangle), \\
 &\vdots\\
d(C_{t}) & d(\langle \theta_{r+1} \rangle \oplus \cdots \oplus \langle \theta_{t} \rangle),\\
d(C_{t+1}) & d(\langle \theta_{r+1} \rangle \oplus \cdots \oplus \langle \theta_{t} \rangle \oplus \langle \theta_{t+1} \rangle), \\
 & \vdots \\
d(C_{s}) & d(\langle \theta_{r+1} \rangle \oplus \cdots \oplus \langle \theta_{t} \rangle \oplus \cdots \oplus\langle \theta_{s}\rangle).
\end{array}$$
Since $d(C_{r+1})=\cdots=d(C_t)=1$ and $d(\langle \theta_{r+1}\rangle)\geq d(\langle \theta_{r+1}\rangle \oplus\langle \theta_{r+2}\rangle)\geq\cdots \geq d(\langle \theta_{r+1} \rangle \oplus \cdots \oplus \langle \theta_{t} \rangle)$, we can shorten the above list as 
\begin{equation}\label{Jensen list}
\setlength\arraycolsep{1pt}
\begin{array}{rl}
& d(\langle \theta_{r+1} \rangle \oplus \cdots \oplus \langle \theta_{t} \rangle),\\
d(C_{t+1}) & d(\langle \theta_{r+1} \rangle \oplus \cdots \oplus \langle \theta_{t} \rangle \oplus \langle \theta_{t+1} \rangle), \\
 &\ \vdots \\
d(C_{s}) & d(\langle \theta_{r+1} \rangle \oplus \cdots \oplus \langle \theta_{t} \rangle \oplus \cdots \oplus\langle \theta_{s}\rangle).
\end{array}
\end{equation}

Now we look into the spectral bound in Theorem \ref{main thm new} in the particular case when $P = \overline{\Omega}$. The related $\lambda$-constacyclic code over $\Fq$ with zero set $\overline{\Omega}$, say $C_{\overline{\Omega}}$, has check polynomial $f(x)=f_{r+1}(x)\cdots f_t(x)$ and by using CRT we get
\begin{equation} \label{largest possible consta}
C_{\overline{\Omega}} = \left\langle\frac{x^m-\lambda}{f(x)}\right\rangle\cong \hspace{-3pt}\bigoplus_{i=r+1}^{t} \F_q[x]/\langle f_i(x)\rangle =\langle \theta_{r+1} \rangle \oplus \ldots \oplus \langle \theta_{t} \rangle.
\end{equation}
Hence, $C_{\overline{\Omega}}=\langle \theta_{r+1} \rangle \oplus \ldots \oplus \langle \theta_{t} \rangle$ is the $\lambda$-constacyclic code over $\Fq$ with the largest possible set $P = \overline{\Omega}$ associated to all eigenvalues. Now we fix $d_{\overline{\Omega}}=d(D_{\overline{\Omega}})$, where $D_{\overline{\Omega}}$ is the $\lambda$-constacyclic code over $\F$ with zero set $\overline{\Omega}$. Using (\ref{common eigencode}) and (\ref{largest possible consta}), we get
\begin{eqnarray}\label{SpecBounds2}
d(C) &\geq&  \min \left\{ d_{\overline{\Omega}}, d(\mathbb{C}_{\overline{\Omega}}) \right\}  \notag\\
&=& \min \left\{ d_{\overline{\Omega}}, d\left(\left(\sum_{i\in\Gamma} \sum_{j=0}^{e_i-1} \overline{C}_i^{q^j}\right) \biggr\rvert_{\Fq}\right) \right\}.
\end{eqnarray} 
Comparing (\ref{SpecBounds2}) with (\ref{Jensen list}) leads to the following cases.

{\bf\noindent	Case 1:}
When $d_{\overline{\Omega}}\leq d(\mathbb{C}_{\overline{\Omega}})$. In this case, the fixed spectral bound (\ref{SpecBounds2}) yields $d(C)\geq d_{\overline{\Omega}}$. Note that $d(C_{\overline{\Omega}})$ is also a member of the list (\ref{Jensen list}), as the topmost element. For an easier representation, let 
$$A_i := d(\langle \theta_{r+1} \rangle \oplus \cdots \oplus \langle \theta_{i} \rangle), \mbox{ for } t+1\leq i \leq s,$$	
so that the collection (\ref{Jensen list}) becomes 
\begin{equation}\label{Jensen list2}
d(C_{\overline{\Omega}}), d(C_{t+1})A_{t+1},\ldots,d(C_{s})A_{s}.
\end{equation}

By definition of $d_{\overline{\Omega}}$, we have $d_{\overline{\Omega}}\leq d(C_{\overline{\Omega}})$. Using Lemma \ref{cons-eig lemma}, we get $d(C_i) \ge d(\overline{C}_i) \geq d(\Bbb{C}_{\overline{\Omega}}) \ge d_{\overline{\Omega}}$, for $t+1\leq i \leq s$. Hence, the estimate of the Jensen bound is greater than or equal to the estimate of the fixed spectral bound in this case. 

{\bf\noindent	Case 2:} 
When $d(\mathbb{C}_{\overline{\Omega}})\leq d_{\overline{\Omega}}$. In this case, the fixed spectral bound (\ref{SpecBounds2}) yields $d(C)\geq d(\mathbb{C}_{\overline{\Omega}})$. Since Lemma \ref{cons-eig lemma} tells us $d(C_i) \ge d(\overline{C}_i) \geq d(\Bbb{C}_{\overline{\Omega}})$, for any $i\in\Gamma$, we also have $d(C_i) A_i \ge d(\Bbb{C}_{\overline{\Omega}})$, for all $t+1\leq i \leq s$. Therefore, the Jensen bound is at least as good as the fixed spectral bound in this case as well.

We have shown the following.
\begin{prop}\label{Jensen win}
Let $C\subseteq R^\ell$ be a $\lambda$-QT code with a nonempty eigenvalue set $\overline{\Omega}\subseteq\Omega$ and let $D_{\overline{\Omega}}$ be the $\lambda$-constacyclic code over $\F$ with zero set $\overline{\Omega}$. Let $d_J$ denote the estimate on $d(C)$ of the Jensen bound. Then, $d(C)\geq d_J \geq \min \left\{d(D_{\overline{\Omega}}), d(\mathbb{C}_{\overline{\Omega}}) \right\}$.
\end{prop}

However, for some choices of $\mathcal{B}(D_{\overline{\Omega}})$, if we consider the associated $d_{Spec}(\mathcal{B}(D_{\overline{\Omega}}))$ that maximizes $\min \left\{ d_P, d(\mathbb{C}_P) \right\}$ over all choices of $(P,d_P)\in\mathcal{B}(D_{\overline{\Omega}})$ with $\emptyset\neq P\subseteq\overline{\Omega}$, then Proposition \ref{Jensen win} above does not hold with $\min \left\{d(D_{\overline{\Omega}}), d(\mathbb{C}_{\overline{\Omega}}) \right\}$ replaced by $d_{Spec}(\mathcal{B}(D_{\overline{\Omega}}))$, as presented in the next example. First, we fix some further notation. Let 
\begin{multline*}
\mathcal{B}_1(D_{\overline{\Omega}}):=\{(P,d(D_P)) : \emptyset \neq P\subseteq \overline{\Omega}, \\
D_P\subseteq\F^m \mbox{\ has\ zero\ set\ } P\},
\end{multline*} 
and let $\mathcal{B}_2(D_{\overline{\Omega}}), \mathcal{B}_3(D_{\overline{\Omega}}), \mathcal{B}_4(D_{\overline{\Omega}})$ denote the BCH, HT and Roos bounds for $D_{\overline{\Omega}}$, respectively, as done in Section \ref{consta sect}. Note that $\mathcal{B}_2(D_{\overline{\Omega}})\subseteq \mathcal{B}_3(D_{\overline{\Omega}})\subseteq \mathcal{B}_4(D_{\overline{\Omega}})$. We also consider their union denoted by 
\begin{eqnarray}\label{big B}
\mathcal{B}_u(D_{\overline{\Omega}})&:=& \mathcal{B}_1(D_{\overline{\Omega}})\cup\mathcal{B}_2(D_{\overline{\Omega}})\cup \mathcal{B}_3(D_{\overline{\Omega}})\cup \mathcal{B}_4(D_{\overline{\Omega}}) \notag\\ &=& \mathcal{B}_1(D_{\overline{\Omega}})\cup \mathcal{B}_4(D_{\overline{\Omega}})
\end{eqnarray}
as an additional choice.

\begin{ex}\label{ex2}
Let $\F_3(\alpha):=\F_{3^6}$ be the splitting field of $x^7+1$. We consider the $[14,7,4]_3~2$-QT code with $(m,\ell,r)=(7,2,1)$, generated by 
$2x^5 + x^3 + x^2 + 2$ and $x^5 + x^3 + x^2$. The corresponding upper-triangular matrix  $\widetilde{G}(x)$ is
\[
\begin{pmatrix}
x + 1  & 2x^5  + 2 x^4 + 2 x^3 + 2\\
0 &  x^6 + 2 x^5 + x^4 + 2 x^3 + x^2 + 2 x + 1
\end{pmatrix},
\]
where $\det(\widetilde{G}(x))=x^7+1$ and, therefore, $\overline{\Omega}=\Omega = \{\alpha^{52}, \alpha^{156}, \alpha^{260}, 2, \alpha^{468}, \alpha^{572}, \alpha^{676}\}$.
Over $\F_3$, the scalar generator matrix is of the form
\[
\setcounter{MaxMatrixCols}{20}
\begin{pmatrix}
1 & 0 & 0 & 0 & 0 & 0 & 2 & 0 & 0 & 2 & 2 & 2 & 0 & 1\\
0 & 1 & 0 & 0 & 0 & 0 & 1 & 0 & 2 & 2 & 2 & 1 & 1 & 0\\
0 & 0 & 1 & 0 & 0 & 0 & 2 & 0 & 0 & 1 & 1 & 1 & 1 & 2\\
0 & 0 & 0 & 1 & 0 & 0 & 1 & 0 & 1 & 0 & 0 & 1 & 2 & 2\\
0 & 0 & 0 & 0 & 1 & 0 & 2 & 0 & 1 & 2 & 0 & 1 & 2 & 2\\
0 & 0 & 0 & 0 & 0 & 1 & 1 & 0 & 1 & 1 & 1 & 0 & 2 & 0\\
0 & 0 & 0 & 0 & 0 & 0 & 0 & 1 & 2 & 1 & 2 & 1 & 2 & 1
\end{pmatrix}.
\]
There are exactly two choices of $P\subseteq\overline{\Omega}$ for which $d_{Spec}(\mathcal{B}_1(D_{\overline{\Omega}}) ; P,d_P)$ yields $d_{Spec}(\mathcal{B}_1(D_{\overline{\Omega}}))$, namely \begin{center}$P=\{\alpha^{52}, \alpha^{468}, \alpha^{572}\}$ and $P=\{\alpha^{156}, \alpha^{260}, \alpha^{676}\}$,\end{center} that give $d_P = d(D_P) = 4$ and $d(\mathbb{C}_P)=\infty$, where $D_P$ is the $2$-constacyclic code of length $7$ with zero set $P$ over $\F_{3^6}$. However, the BCH-like, the HT-like, and the Roos-like bounds for QT codes, namely $d_{Spec}(\mathcal{B}_2(D_{\overline{\Omega}}))$, $d_{Spec}(\mathcal{B}_3(D_{\overline{\Omega}}))$ and $d_{Spec}(\mathcal{B}_4(D_{\overline{\Omega}}))$, are not sharp, for any nonempty $P\subseteq\overline{\Omega}$, where they all yield $d_P=|P| + 1 =3$ and $d(\mathbb{C}_{P})=\infty$, for $P=\{\alpha^{52}, \alpha^{468}\}$, as their best estimate. Hence, we have $d_{Spec}(\mathcal{B}_u(D_{\overline{\Omega}}))=d_{Spec}(\mathcal{B}_1(D_{\overline{\Omega}}))=4$, whereas $d_{J}=2$.
\end{ex}

\begin{rem}
Taking a sufficiently large $P$ as well as considering different choices of defining set bounds is crucial when using the spectral bound. For instance, for any fixed choice of defining set bound $\mathcal{B}(D_{\overline{\Omega}})$, if we consider the behaviour of $d_{Spec}(\mathcal{B}(D_{\overline{\Omega}}) ; P, d_P)$ over the nonempty subsets $P\subseteq \overline{\Omega}$ with $(P,d_P)\in\mathcal{B}(D_{\overline{\Omega}})$, then $d_P$ might be nondecreasing as $P$ approaches $\overline{\Omega}$, whereas $d(\mathbb{C}_P)$ is nonincreasing, which follows from (\ref{common eigencode}) since we might add more to the sum. Hence, the optimized value $d_{Spec}(\mathcal{B}(D_{\overline{\Omega}}))$ might not be obtained by involving all eigenvalues, unlike the case of constacyclic codes, as any nonconstant $d_P$ grows proportionately to the size of $P$. 
\end{rem}

\begin{ex}\label{ex 5}
Let us consider the same $[14,7,4]_3~2$-QT code in Example \ref{ex2} above with eigenvalues $$\overline{\Omega}=\Omega = \{\alpha^{52}, \alpha^{156}, \alpha^{260}, 2, \alpha^{468}, \alpha^{572}, \alpha^{676}\}.$$ It is obvious that any subset of $\overline{\Omega}$ with 1 or 2 elements is consecutive. We first take $P=\{\alpha^{52}\}$ and use the BCH-like bound, where we obtain $d_P=2$ and $d(\mathbb{C}_{P})=\infty$, and hence $d_{Spec}(\mathcal{B}_2(D_{\overline{\Omega}}) ; P,d_P)=2$ in this case. If we add one more element and consider $P=\{2, \alpha^{52}\}$, then we get $d_P=3$ but $d(\mathbb{C}_{P})=1$, which makes $d_{Spec}(\mathcal{B}_2(D_{\overline{\Omega}}) ; P,d_P)=1$. If we replace $2$ by $\alpha^{468}$ and take $P=\{\alpha^{52}, \alpha^{468}\}$ instead, then we obtain $d_{Spec}(\mathcal{B}_2(D_{\overline{\Omega}}) ; P,d_P)=d_{Spec}(\mathcal{B}_2(D_{\overline{\Omega}}))=3$, as indicated in Example \ref{ex2} above.
\end{ex}

\subsection{Lally against Spectral} \label{LS sect}

Lally derived a lower bound on the minimum distance of a given QC code in \cite{L2}. Her findings can be adapted to QT codes easily and the same minimum distance bound can be extended to the QT codes in an analogous way.

Let $C$ be a $\lambda$-QT code of length $m\ell$ and index $\ell$ over $\Fq$. Let $\{1, \gamma, \ldots, \gamma^{\ell-1}\}$ be some fixed choice of basis of $\F_{q^\ell}$ as a vector space over $\Fq$. We view the codewords of $C$ as $m\times \ell$ arrays as in (\ref{array}) and consider the following map:
\begin{equation*}\begin{array}{lll} 
\tau : \hspace{2cm} \F_q^{m\ell} & \longrightarrow & \hspace{.6cm}\F_{q^\ell}^m  \\[8pt]
\mathbf{c}=\left(
  \begin{array}{ccc}
    c_{0,0} & \ldots & c_{0,\ell-1} \\
    \vdots &  & \vdots \\
    c_{m-1,0} & \ldots & c_{m-1,\ell-1} \\
  \end{array}
\right) & \longmapsto &   \left(\begin{array}{c}
    c_{0} \\
    \vdots \\
    c_{m-1} \\
  \end{array}\right) ,
\end{array}
\end{equation*}
where $c_i=c_{i,0} + c_{i,1} \, \gamma + \cdots + c_{i,\ell-1} \, \gamma^{\ell-1}\in \F_{q^\ell}$, for all $0\leq i \leq m-1$.

Clearly, $\tau(\mathbf{c})$ lies in some $\lambda$-constacyclic code, for any $\mathbf{c}\in C$. We now define the smallest such constacyclic code as $\widehat{C}$, which contains all of $\tau(C)$. First, we equivalently extend the map $\tau$ above to the polynomial description of codewords as (cf. (\ref{identification-1}))
\begin{eqnarray}\label{Lal}
\tau : \Fq[x]^{\ell} &\To& \F_{q^{\ell}}[x] \\
\mathbf{c}(x)=(c_0(x),\ldots,c_{\ell-1}(x)) &\longmapsto& c(x)=\displaystyle\sum_{j=0}^{\ell-1}c_j(x) \, \gamma^j .\notag
\end{eqnarray}
If $C$ has generating set $\{\mathbf{f}_1,\ldots,\mathbf{f}_r\}$, where 
\[
\mathbf{f}_k := (f^{(k)}_{0}(x),\ldots,f^{(k)}_{\ell-1}(x))\in \Fq[x]^{\ell},
\]
for each $k \in \{1,\ldots,r\}$, then $$\widehat{C}=\langle \gcd(f_1(x),\ldots, f_r(x), x^m-\lambda)\rangle$$ such that $f_k=\tau(\mathbf{f}_k)\in \F_{q^{\ell}}[x]$, for all $k$ \cite{L2}.

Next, we consider the $q$-ary linear code of length $\ell$ that is generated by the rows of the codewords in $C$, which are represented as $m\times \ell$ arrays as in (\ref{array}). Recall that this code was denoted by $B$ in the proof of Lemma \ref{freaking lemma}. Namely, $B$ is the linear block code of length $\ell$ over $\Fq$, generated by $\{\mathbf{f}_{k,i} \, : \, k \in \{1,\ldots,r\},\ i \in \{0,\ldots,m-1\}\} \subseteq \F_{q}^{\ell}$, where each $\mathbf{f}_{k,i}:=(f^{(k)}_{i,0},\ldots,f^{(k)}_{i,\ell-1}) \in\Fq^\ell$ is the vector of the $i^{\rm th}$ coefficients of the polynomials 

\[
f^{(k)}_{j} (x)=f^{(k)}_{0,j}+f^{(k)}_{1,j}x+\cdots+f^{(k)}_{m-1,j}x^{m-1},
\]
for all $k \in \{1,\ldots,r\}$ and $j \in \{0,\ldots,\ell-1\}$.

Since the image of any codeword $\mathbf{c}(x) \in C$ under the map $\tau$ is an element of the $\lambda$-constacyclic code $\widehat{C}$ over $\F_{q^\ell}$, there are at least $d(\widehat{C})$ nonzero rows in each nonzero codeword of $C$. For any $i \in \{0,\ldots,m-1\}$, the $i^{\rm th}$ row $\mathbf{c}_i=(c_{i,0},\ldots,c_{i,\ell-1})$ of any $\mathbf{c}\in C$ can be viewed as a codeword in $B$, therefore, a nonzero $\mathbf{c}_i$ has weight at least $d(B)$. Hence, we have shown the following.

\begin{thm}(cf.~\cite[Theorem 5]{L2})
Let $C$ be an $r$-generator $\lambda$-QT code of length $m\ell$ and index $\ell$ over $\F_{q}$ with generating set $\{\mathbf{f}_1,\ldots,\mathbf{f}_r\}\subseteq \Fq[x]^{\ell}$. Let the $\lambda$-constacyclic code $\widehat{C}\subseteq \F_{q^\ell}^m$ and the linear code $B\subseteq\Fq^{\ell}$ be defined as above. Then, we have \[d(C)\geq d(\widehat{C}) \, d(B).\] 
\end{thm}

We now prove that an analogue of Proposition \ref{Jensen win} holds between the Lally and spectral bounds, which is an immediate result of Lemma \ref{freaking lemma}.

\begin{cor}\label{Lally win}
Let $C\subseteq R^\ell$ be a nontrivial $\lambda$-QT code with the eigenvalue set $\overline{\Omega}=\Omega$ and let $d_L$ denote the estimate on $d(C)$ of the Lally bound. Then, $d(C)\geq d_L \geq d(\mathbb{C}_{\Omega})$. In particular, $d_L \geq \min\{d_{\Omega}, d(\mathbb{C}_{\Omega})\}$, for any $d_{\Omega}\geq 1$.
\end{cor}

\begin{proof}
By Lemma \ref{freaking lemma}, we already know that $B$ is a subcode of $\mathbb{C}_{\Omega}$ and therefore we have $d(B) \geq d(\mathbb{C}_{\Omega})$. Hence, $d_L=d(\widehat{C})d(B) \geq d(\mathbb{C}_{\Omega}) \geq \min\{d_{\Omega}, d(\mathbb{C}_{\Omega})\}$.
\end{proof}

However, Corollary \ref{Lally win} does not apply to the case when $d_{Spec}(\mathcal{B}(D_{\Omega}))$ is considered, instead of some fixed choice for $\min\{d_{\Omega}, d(\mathbb{C}_{\Omega})\}$, as highlighted by the codes in the next two examples.

\begin{table*}[t!]
\caption{The outcomes on the performance comparison of the bounds. We list the number of instances, with double counting allowed, when a specified bound reached the actual minimum distance, \textit{i.e.}, it was sharp, and when it was greater than or equal to the other two bounds, \textit{i.e.}, it was best-performing.}
\label{table:plot1}
\renewcommand{\arraystretch}{1.2}
\centering
\begin{tabular}{c | c |c|c | c|c|c | c|c|c}
\hline
 & \multicolumn{3}{c|}{QC for $q=2$}  & \multicolumn{3}{c|}{QC for $q=3$} & \multicolumn{3}{c}{QT for $q=3, \lambda=2$} \\ 
\hline
bound &  $d_{L}$ & $d_{Spec}$ & $d_J$ &  $d_{L}$ & $d_{Spec}$ & $d_J$ &  $d_{L}$ & $d_{Spec}$ & $d_J$ \\ 
\hline 
sharp & $16\,563$ & $27\,438$ & $31\,368$ & $14\,704$ & $31\,711$ & $33\,799$ & $4\,037$ & $9\,257$ & $9\,780$ \\
best-performing & $20\,821$ & $61\,012$ & $92\,506$ & $19\,359$ & $83\,284$ & $143\,010$ & $5\,041$ & $35\,363$ & $65\,275$ \\
\hline
\# nontrivial $C$ & \multicolumn{3}{c|}{$93\,467$} & \multicolumn{3}{c|}{$143\,602$} & \multicolumn{3}{c}{$65\,589$} \\
\hline 
\end{tabular} 
\end{table*}

\begin{table*}[t!]
\caption{The outcomes on the comparison of the bounds in strictly decreasing patterns. The label of being sharp applies only to the bound with the largest value among the three.}
\label{table:plot2}
\renewcommand{\arraystretch}{1.2}
\centering
\begin{tabular}{c | c |c|c|c | c|c|c|c | c|c|c|c}
\hline
 & \multicolumn{4}{c|}{QC for $q=2$}  & \multicolumn{4}{c|}{QC for $q=3$} & \multicolumn{4}{c}{QT for $q=3, \lambda=2$} \\ 
\hline
decreasing order & JSL & JLS & SJL & LJS &  JSL & JLS & SJL & LJS &  JSL & JLS & SJL & LJS \\ 
\hline 

count & 
$17\,976$ & $77$ & $398$ & $11$ & 

$21\,622$ & $25$ & $102$ & $1$ &

$22\,475$ & $3$ & $165$ & $1$ \\

sharp    & 
$3\,203$ & $69$ & $94$ & $4$ & 

$965$ & $9$ & $24$ & $0$ &

$1\,030$ & $0$ & $26$ & $0$ \\

\hline

\# nontrivial $C$ & \multicolumn{4}{c|}{$56\,243$} & \multicolumn{4}{c|}{$53\,004$} & \multicolumn{4}{c}{$52\,915$} \\
	\hline 
\end{tabular} 
\end{table*}

\begin{ex}\label{ex4}
Let $\F_3(\alpha) := \F_{81}$ be the splitting field of $x^{10}+1$. Consider the $[20, 10, 4]_3$ 2-QT code with $(m,\ell,r) = (10,2,1)$, generated by $2x^8 + 2x^7 + x^2 + x + 1$ and $2x^6 + 2x^5 + x^4 + x + 2$. The associated upper-triangular matrix  $\widetilde{G}(x)$ is
\[
\begin{pmatrix}
1  & 2x^9 + 2x^7 + 2x^6 + x^5 + 2x^3 + x^2 + 1\\
0 &  x^{10} +1
\end{pmatrix},
\]
where $\det(\widetilde{G}(x))=x^{10}+1$, which implies $\overline{\Omega}=\Omega$.
Over $\F_3$, the scalar generator matrix is 
\[
\setcounter{MaxMatrixCols}{20}\setlength\arraycolsep{3.5pt}
\begin{pmatrix}
1 & 0 & 0 & 0 & 0 & 0 & 0 & 0 & 0 & 0 & 1 & 0 & 1 & 2 & 0 & 1 & 2 & 2 & 0 & 2\\
0 & 1 & 0 & 0 & 0 & 0 & 0 & 0 & 0 & 0 & 1 & 1 & 0 & 1 & 2 & 0 & 1 & 2 & 2 & 0\\
0 & 0 & 1 & 0 & 0 & 0 & 0 & 0 & 0 & 0 & 0 & 1 & 1 & 0 & 1 & 2 & 0 & 1 & 2 & 2\\
0 & 0 & 0 & 1 & 0 & 0 & 0 & 0 & 0 & 0 & 1 & 0 & 1 & 1 & 0 & 1 & 2 & 0 & 1 & 2\\
0 & 0 & 0 & 0 & 1 & 0 & 0 & 0 & 0 & 0 & 1 & 1 & 0 & 1 & 1 & 0 & 1 & 2 & 0 & 1\\
0 & 0 & 0 & 0 & 0 & 1 & 0 & 0 & 0 & 0 & 2 & 1 & 1 & 0 & 1 & 1 & 0 & 1 & 2 & 0\\
0 & 0 & 0 & 0 & 0 & 0 & 1 & 0 & 0 & 0 & 0 & 2 & 1 & 1 & 0 & 1 & 1 & 0 & 1 & 2\\
0 & 0 & 0 & 0 & 0 & 0 & 0 & 1 & 0 & 0 & 1 & 0 & 2 & 1 & 1 & 0 & 1 & 1 & 0 & 1\\
0 & 0 & 0 & 0 & 0 & 0 & 0 & 0 & 1 & 0 & 2 & 1 & 0 & 2 & 1 & 1 & 0 & 1 & 1 & 0\\
0 & 0 & 0 & 0 & 0 & 0 & 0 & 0 & 0 & 1 & 0 & 2 & 1 & 0 & 2 & 1 & 1 & 0 & 1 & 1
\end{pmatrix}.
\] 
There are two choices of $P$, namely \begin{center}$P=\{\alpha^{52}, \alpha^{60}, \alpha^{68}\}$ and $P=\{\alpha^{20}, \alpha^{44}, \alpha^{76}\}$,\end{center} that yield $d_{Spec}(\mathcal{B}_1(D_{\overline{\Omega}}) ; P,d_P)=d_{Spec}(\mathcal{B}_1(D_{\overline{\Omega}}))$ with $d_P= d(D_P)= 4$ and $d(\mathbb{C}_P)=\infty$, where $D_P$ is the $2$-constacyclic code of length $10$ with zero set $P$ over $\F_{81}$. Using the first choice of $P$, which is consecutive, the BCH-like --and therefore the HT-like and the Roos-like-- bounds for QT codes are also sharp, giving $d_P=4$ and $d(\mathbb{C}_{P})=\infty$. Hence, we obtain $d_{Spec}(\mathcal{B}_u (D_{\overline{\Omega}}))=4=d_{Spec}(\mathcal{B}_i (D_{\overline{\Omega}}))$, for any $1\leq i\leq 4$.  Using the second choice of $P$, which appears in $\mathcal{B}_3(D_{\overline{\Omega}})$ and $\mathcal{B}_4(D_{\overline{\Omega}})$ but not in $\mathcal{B}_2(D_{\overline{\Omega}})$, the HT-like and the Roos-like bounds for QT codes are sharp with $d_P=4$ and $d(\mathbb{C}_{P})=\infty$. Therefore, we have $d_{Spec}(\mathcal{B}_u (D_{\overline{\Omega}}))=4=d_{Spec}(\mathcal{B}_i (D_{\overline{\Omega}}))$, for $ i\in \{1,3,4\}$, whereas $d_L = 1$ and $d_{J}=2$.
\end{ex}

\begin{ex}\label{ex3}
Let $\F_3(\alpha) := \F_9$ be the splitting field of $x^8+2$. Consider the $[16,7,5]_3$ QC code with $(m,\ell,r) = (8,2,1)$, generated by $x^7 + x^6 + x^5 + x^4 + 2x^3 +2$ and $x^7 + x^5 + 2 x^4 + 2 x^3 + x$. The associated upper-triangular matrix  $\widetilde{G}(x)$ is
\[
\begin{pmatrix}
x + 1  & x^6 + x^5 + 2x^4 + 2 x^3 + 2 x^2 +2 x\\
0 &  x^8 +2
\end{pmatrix},
\]
where $\det(\widetilde{G}(x))=(x+1)(x^8+2)$ and therefore $\overline{\Omega}=\Omega = \F_9\setminus \{0\}$.

Over $\F_3$, the scalar generator matrix is of the form
\[
\setcounter{MaxMatrixCols}{20}
\begin{pmatrix}
1 & 0 & 0 & 0 & 0 & 0 & 0 & 1 & 2 & 2 & 2 & 2 & 1 & 1 & 0 & 0 \\
0 & 1 & 0 & 0 & 0 & 0 & 0 & 2 & 1 & 0 & 0 & 0 & 1 & 0 & 1 & 0 \\
0 & 0 & 1 & 0 & 0 & 0 & 0 & 1 & 2 & 0 & 2 & 2 & 1 & 2 & 0 & 1 \\
0 & 0 & 0 & 1 & 0 & 0 & 0 & 2 & 2 & 0 & 1 & 0 & 1 & 0 & 2 & 0 \\
0 & 0 & 0 & 0 & 1 & 0 & 0 & 1 & 2 & 1 & 2 & 0 & 1 & 2 & 0 & 2 \\
0 & 0 & 0 & 0 & 0 & 1 & 0 & 2 & 0 & 0 & 2 & 0 & 2 & 0 & 2 & 0 \\
0 & 0 & 0 & 0 & 0 & 0 & 1 & 1 & 2 & 2 & 2 & 1 & 1 & 0 & 0 & 2
\end{pmatrix}.
\] 
There are again two choices of $P$, namely \begin{center}$P=\{\alpha^2, \alpha^3, 2, \alpha^5\}$ and $P=\{\alpha, 2, \alpha^6, \alpha^7\}$,\end{center} that yield $d_{Spec}(\mathcal{B}_1(D_{\overline{\Omega}}) ; P,d_P)=d_{Spec}(\mathcal{B}_1(D_{\overline{\Omega}}))$ with $d_P=	d(D_P)= 5$ and $d(\mathbb{C}_P)=\infty$, where $D_P$ is the cyclic code of length $8$ with zero set $P$ over $\F_{9}$. Using the first choice of $P$, which is clearly consecutive, the BCH-like, the HT-like and the Roos-like bounds for QC codes are also sharp with $d_P=5$ and $d(\mathbb{C}_{P})=\infty$. Hence, we get $d_{Spec}(\mathcal{B}_u (D_{\overline{\Omega}}))=5=d_{Spec}(\mathcal{B}_i (D_{\overline{\Omega}}))$, for $1\leq i\leq 4$. The second choice of $P$ is not consecutive but it appears in $\mathcal{B}_3(D_{\overline{\Omega}})$ and $\mathcal{B}_4(D_{\overline{\Omega}})$, making the HT-like and the Roos-like bounds again sharp with $d_P=5$ and $d(\mathbb{C}_{P})=\infty$. Therefore, we obtain $d_{Spec}(\mathcal{B}_u (D_{\overline{\Omega}}))=5=d_{Spec}(\mathcal{B}_i (D_{\overline{\Omega}}))$, for $ i\in \{1,3,4\}$, whereas $d_L = 2$ and $d_{J}=4$.
\end{ex}

\begin{figure*}[ht!]
\centering
\caption{Visualization, in percentage, of the outcomes listed in Tables~\ref{table:plot1} and~\ref{table:plot2}.}
\begin{tabular}{cc}
\includegraphics[width=0.47\linewidth]{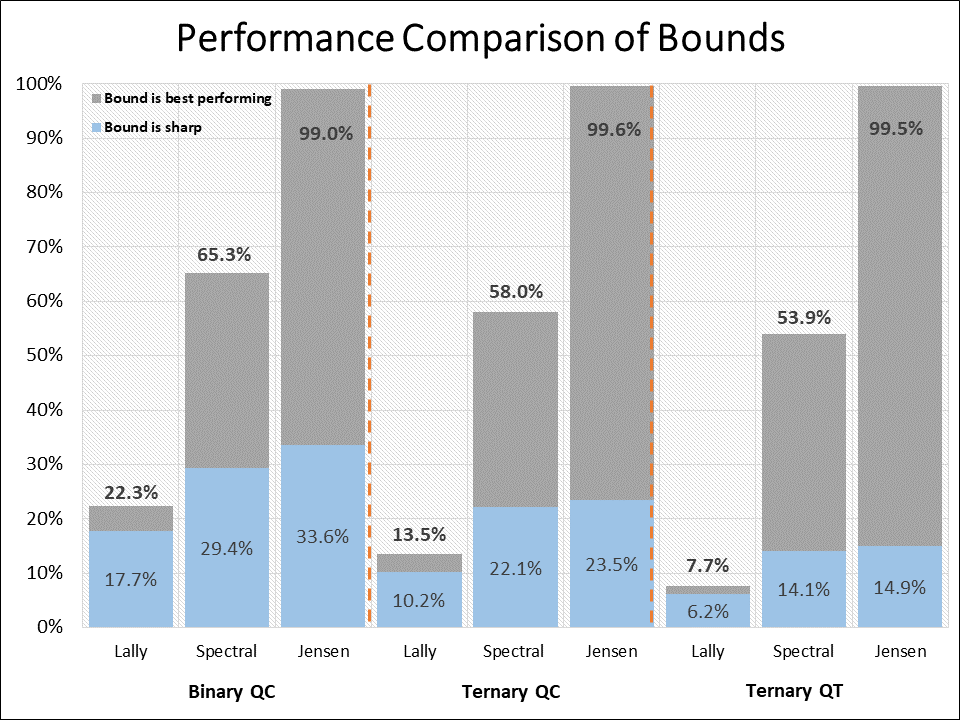} & 
\includegraphics[width=0.47\linewidth]{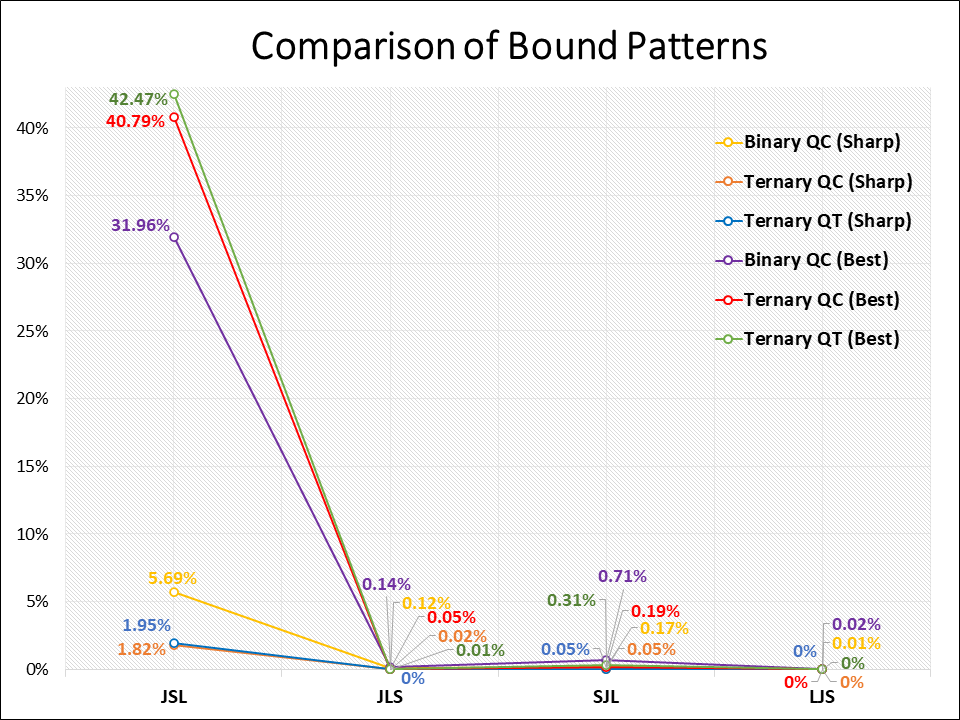} \\
(a) Overall performance. & (b) Performance by patterns. 
\end{tabular}
\label{fig:comparison}
\end{figure*}

\section{Numerical Comparisons}\label{Example section}

To compare the performance of the bounds we carried out two procedures. The first one looks into the overall performance of the bounds. The second one investigates their strict ranking.

For $q=2$, we construct $1000$ random codes on each input tuple $(m,\ell,r)$, with $m \in \{3,5,7,9,11\}$, $2 \leq \ell \leq 6$, and $1 \leq r \leq \ell$. Once $(m,\ell,r)$ is fixed, an array $\mathcal{A}$ of generator polynomials is randomly built. The number of polynomials in this array is $r \ell$. The corresponding binary $r$-generator QC code $C$ is generated by using the {\tt QuasiCyclicCode} function in \textsc{magma} with input $(m \ell, \mathcal{A}, r)$. If $C$ was nontrivial, then its minimum distance $d(C)$ and the values given by the three bounds $d_{L}$, $d_{Spec}$, and $d_{J}$ were determined, where $d_{Spec} :=d_{Spec}(\mathcal{B}_u(D_{\overline{\Omega}}))$ (see (\ref{big B})). We recorded the respective numbers of occasions when each bound was either sharp or was best-performing among the three. Double counting was allowed in cases where two or more bounds were simultaneously sharp. Similarly, double counting was also allowed for coinciding estimates with the best-performing bounds.

An identical routine was carried out for $q=3$. To keep $\gcd(m,q)=1$, we used $m \in \{4,5,7,8\}$, with $2 \leq \ell \leq 6$ and $1 \leq r \leq \ell$. When $\lambda=1$, we constructed $2000$ random codes on each input tuple $(m,\ell,r)$. Changing $\lambda$ to $2$, we completed $1000$ random constructions on each input tuple. In the strictly QT setup, the {\tt QuasiTwistedCyclicCode} function in \textsc{magma} generated the ternary $2$-QT code $C$ on input $(m \ell, \mathcal{A}, 2)$, built from the randomly generated polynomials in the corresponding array $\mathcal{A}$ of length $r \ell$. The outcomes of the first procedure can be found in Table~\ref{table:plot1} with the visualization given in Plot (a) of Figure~\ref{fig:comparison}.

In the second procedure we recorded the number of random constructions in which the bounds ranked in strictly decreasing order. We label the $6$ possible patterns in abbreviated form. The pattern JSL stands for $d_{J} > d_{Spec} > d_L$. The other patterns are analogously interpreted. Two patterns, namely SLJ and LSJ, for $d_{Spec} > d_{L} > d_{J}$ and $d_{L} >d_{Spec} > d_{J}$, respectively, never occurred in our constructions. For $q=2$, we ran $1000$ random codes on each input $(m,\ell,r)$, with $m \in \{3,5,7\}$, 
$2 \leq \ell \leq 6$ and $1 \leq r \leq \ell$. For $q=3, \lambda=1$ and $q=3, \lambda=2$, we used $m \in \{4,5,7,8\}$. The exact counts are presented in Table~\ref{table:plot2} and interpreted visually in Plot (b) of Figure~\ref{fig:comparison}. 

Our random constructions reveal that the Jensen bound has the best overall performances by a wide percentage margin. The situation where $d_{J} \geq d_{Spec} \geq d_{L}$ is typical. Consider, for example, the $[8, 2, 6]_3~2$-QT code with $(m,\ell,r) = (4,2,1)$, generated by $2x^3 + 2x + 1$ and $2x^2 + x + 1$. The code is optimal in terms of minimum distance with $d(C) = d_{J} = 6 > d_{Spec}=4 > d_{L}=3$. It has a generator matrix
\[
\begin{pmatrix}
1 & 0 & 1 & 1 & 2 & 1 & 0 & 1\\
0 & 1 & 1 & 2 & 1 & 0 & 1 & 1
\end{pmatrix}.
\]
While it is almost always safe to bet on the Jensen bound, there are occasions where $d_L$, similarly $d_{Spec}$, outperforms the other two bounds. 

\begin{figure*}[t!]
	\centering
	\caption{Patterns showing the ratios of the distance bound estimates given by $\frac{d_J}{d}, \frac{d_{Spec}}{d}, \frac{d_L}{d}$ in terms of the rate of the codes. The dotted lines are their respective quadratic curve fitting lines.}
	\begin{tabular}{cc}
		\includegraphics[width=0.46\linewidth]{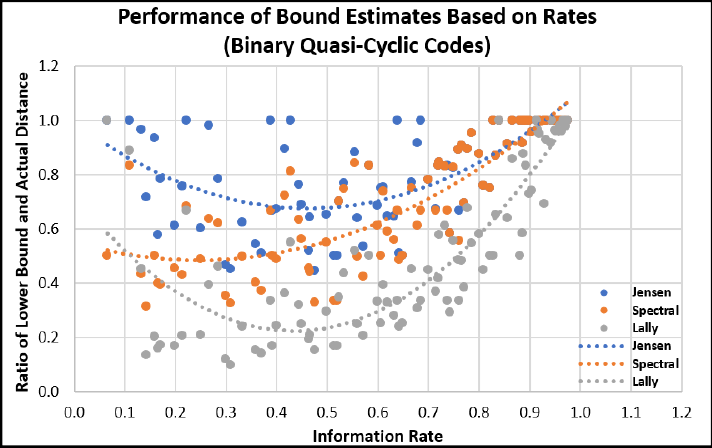} & 
		\includegraphics[width=0.46\linewidth]{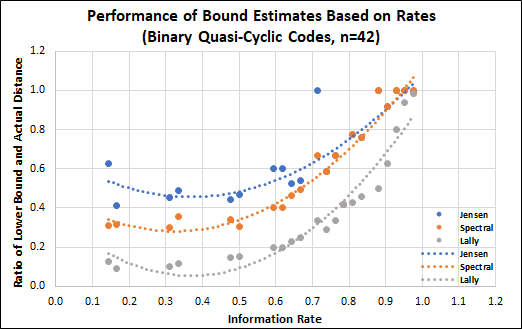} \\
		(a) Binary QC & (b) Binary QC of fixed length $42$\\
		\includegraphics[width=0.46\linewidth]{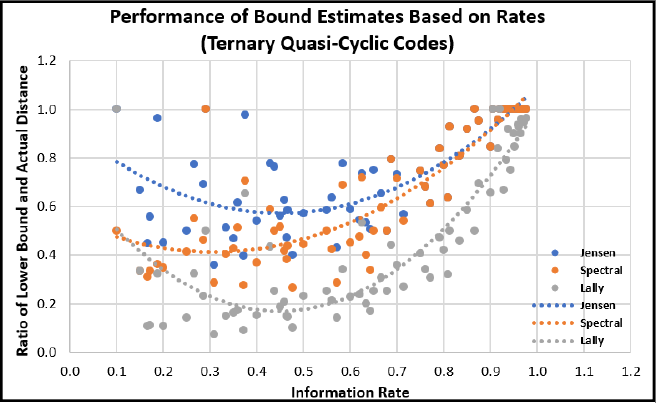} & 
		\includegraphics[width=0.46\linewidth]{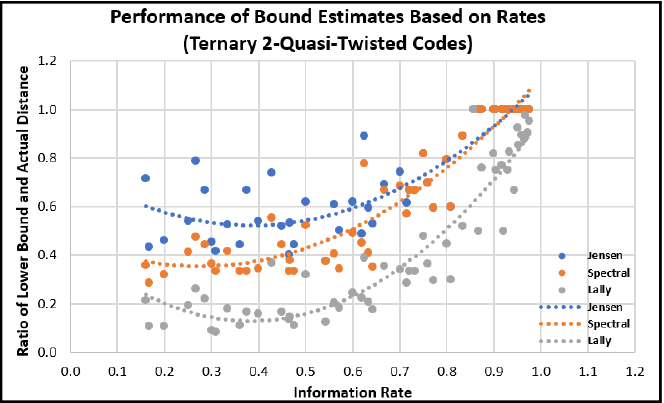} \\
		(c) Ternary QC & (d) Ternary QT with $\lambda=2$ 
	\end{tabular}
	\label{fig:pattern}
\end{figure*}

\begin{ex}\label{ex 4}
We start with an example where $d_{L} > d_{J} > d_{Spec}$ and the code is optimal since it reaches the best-possible minimum distance. The $[21, 6, 8]_2$ QC code with $(m,\ell,r)=(7,3,1)$, generator polynomials $x^4 + x^3 + x + 1$, $x^4 + 1$, and $x^3 + x$, and generator matrix
\[
\setcounter{MaxMatrixCols}{30}\setlength\arraycolsep{3pt}
\begin{pmatrix}
1 & 0 & 0 & 0 & 0 & 0 & 1 & 1 & 1 & 1 & 0 & 1 & 1 & 1 & 0 & 1 & 1 & 0 & 1 & 1 & 0 \\
0 & 1 & 0 & 0 & 0 & 0 & 1 & 0 & 0 & 0 & 1 & 1 & 0 & 0 & 0 & 1 & 0 & 1 & 1 & 0 & 1 \\
0 & 0 & 1 & 0 & 0 & 0 & 1 & 1 & 1 & 1 & 0 & 0 & 0 & 1 & 1 & 1 & 0 & 0 & 0 & 0 & 0 \\
0 & 0 & 0 & 1 & 0 & 0 & 1 & 0 & 0 & 0 & 1 & 1 & 1 & 1 & 0 & 0 & 0 & 0 & 1 & 1 & 0 \\
0 & 0 & 0 & 0 & 1 & 0 & 1 & 0 & 1 & 1 & 0 & 0 & 0 & 0 & 0 & 1 & 1 & 0 & 1 & 0 & 1 \\
0 & 0 & 0 & 0 & 0 & 1 & 1 & 1 & 1 & 0 & 1 & 1 & 1 & 1 & 1 & 1 & 0 & 1 & 1 & 0 & 0
\end{pmatrix},
\]
has $d_{L}=8 > d_{J} = 6 > d_{Spec}=4$.

Here is an example where $d_{Spec} > d_{J} > d_{L}$. In the $[21, 7, 6]_2$ {QC code} with $(m,\ell,r)=(7,3,1)$, generated by $x^6 + x^5 + x^4 + x^2 + 1$, $x^6 + x^5 + x^2$, and 
$x^6 + x^5 + x^4 + 1$, none of the bounds is sharp, since $d(C)=6 > d_{Spec} =5 > d_{J} = 3 > d_{L}=1$. The code has a generator matrix
\[
\setcounter{MaxMatrixCols}{30}\setlength\arraycolsep{3pt}
\begin{pmatrix}
1 & 0 & 0 & 0 & 0 & 0 & 0 & 1 & 1 & 0 & 0 & 0 & 0 & 1 & 0 & 0 & 0 & 1 & 1 & 0 & 0\\
0 & 1 & 0 & 0 & 0 & 0 & 0 & 1 & 1 & 1 & 0 & 0 & 0 & 0 & 0 & 0 & 0 & 0 & 1 & 1 & 0\\
0 & 0 & 1 & 0 & 0 & 0 & 0 & 0 & 1 & 1 & 1 & 0 & 0 & 0 & 0 & 0 & 0 & 0 & 0 & 1 & 1\\
0 & 0 & 0 & 1 & 0 & 0 & 0 & 0 & 0 & 1 & 1 & 1 & 0 & 0 & 1 & 0 & 0 & 0 & 0 & 0 & 1\\
0 & 0 & 0 & 0 & 1 & 0 & 0 & 0 & 0 & 0 & 1 & 1 & 1 & 0 & 1 & 1 & 0 & 0 & 0 & 0 & 0\\
0 & 0 & 0 & 0 & 0 & 1 & 0 & 0 & 0 & 0 & 0 & 1 & 1 & 1 & 0 & 1 & 1 & 0 & 0 & 0 & 0\\
0 & 0 & 0 & 0 & 0 & 0 & 1 & 1 & 0 & 0 & 0 & 0 & 1 & 1 & 0 & 0 & 1 & 1 & 0 & 0 & 0
\end{pmatrix}.
\]
\end{ex}

Figure~\ref{fig:pattern} illustrates the situation more comprehensively. We average, over a large number of distinct nontrivial codes, the respective ratios of distance estimates $d_J, d_{Spec},d_L$ over the actual minimum distance $d$ for fixed code rate $k/n$. For $q=2$, based on $4\,937$ distinct nontrivial QC codes, Figure~\ref{fig:pattern} (a) presents the \emph{average ratios} 
\[
\frac{d_J}{d}, \quad \frac{d_{Spec}}{d}, \quad \frac{d_L}{d}
\]	
on the vertical axis for the indicated code rates along the horizontal axis. The dotted lines are the quadratic curve fitting lines for the respective distance estimates.  Limiting the analysis to a specific length, say $n=42$ as in Figure~\ref{fig:pattern} (b), reveals similar patterns of behaviour. This is the main reason behind our choice of using the rate as reference without fixing the length $n$. The respective plots in Figure~\ref{fig:pattern} (c) and (d) show the patterns based on $37\,063$ distinct nontrivial ternary QC codes and $28\,853$ distinct nontrivial ternary QT codes with $\lambda=2$. The average values in Figure~\ref{fig:pattern} must of course be interpreted alongside the frequencies shown in Figure~\ref{fig:comparison}. 

\section{Conclusion}

We conclude the paper by restating the key insights that we have gained from comparing distance bounds for quasi-twisted codes. 
	
The general spectral bound, presented as Theorem \ref{main thm new}, is encompassing, counting the BCH-like, the HT-like, the Roos-like, and the shift bounds as special cases. It is also intriguingly simple to analyze. Once we have identified a nonempty set of eigenvalues of a given $\lambda$-QT code, we can define two parameters based on any chosen nonempty subset of this set. The first parameter is the minimum distance of the eigencode, which corresponds to the common eigenspace of the eigenvalues in the chosen subset. The second parameter is a minimum distance bound for any $\lambda$-constacyclic code whose zero set contains the chosen subset. The general spectral bound is computed based on these two derived parameters, each of which is simpler to work on than the original QT code.  

In terms of performance, we have seen by numerical comparisons that the bound to beat is overwhelmingly the Jensen bound. Despite its relatively poor performance against the Jensen bound, the general spectral bound may provide a better estimate for QT codes with high dimension. The two parameters of the spectral bound are bounded above by the index and the co-index of the given QT code, whose length is the product of these two. On the other hand, a QT code with a high dimension clearly has many nonzero constituents, which affects the direct sums and the size of the list (\ref{Jensen list}) in the Jensen bound. Hence, its performance starts to decay as the dimension increases whereas the spectral bound gets better, which can be observed in Figure \ref{fig:pattern}. Therefore, it is a good idea to check the estimates of both bounds if the dimension of the given QT code is sufficiently large. The majority of the occasions where the Jensen bound still performs better than the general spectral bound occurs when the QT code with high rate has many full space constituents. In that case, the size of (\ref{Jensen list}) gets smaller with longer direct sums. On the other hand, as indicated in Section \ref{comparison section}, full space constituents yield a smaller number of eigenvalues. Hence, the general spectral bound might be a better choice than the Jensen bound for QT codes of high dimension, unless they have a high number of full space constituents.

In addition to providing QT analogues of the minimum distance bounds given for constacyclic codes, the general spectral bound offers another major theoretical value. It reveals properties that allow for a direct structural comparison with both the Jensen and the Lally bounds. They serve as links, previously unavailable in the literature, that explicitly connect the three bounds. The poor performance of the Lally bound might be improved by extending its single-layer concatenated structure to a more general setup with some multi-layer concatenation, like in the Jensen bound.

\vfill


\begin{thebibliography}{99}

\bibitem {AGOSS1} A. Alahmadi, C. G\"uneri, B. \"Ozkaya, H. Shoaib and P. Sol\'{e}, ``On self-dual double negacirculant codes", \emph{Discrete Appl. Math.}, vol. 222, pp. 205--212, 2017.

\bibitem {AGOSS2} A. Alahmadi, C. G\"uneri, B. \"Ozkaya, H. Shoaib and P. Sol\'{e}, ``On complementary-dual multinegacirculant codes", \emph{Cryptogr. Commun.}, vol. 12, pp. 101--113, 2020.

\bibitem{Barbier} M. Barbier, C. Chabot and G. Quintin, ``On quasi-cyclic codes as a generalization of cyclic codes",
\emph{Finite Fields Appl.}, vol. 18, pp. 904--919, 2012.

\bibitem {B} J. Bierbrauer, ``Introduction to Coding Theory", \emph{Chapman and Hall/CRC Press}, 2016.

\bibitem {BC} R. C. Bose and D. K. R. Chaudhuri, ``On a class of error correcting binary group code", \emph{Inf. Control}, vol. 3, no. 1, pp. 68--79, 1960.

\bibitem {BCP} W. Bosma, J. Cannon and C. Playoust, ``The Magma algebra system. I. The user language", \emph{J. Symbolic Comput.}, vol. 24, pp. 235--265, 1997.

\bibitem{BGS}  J. J. Bernal , M. Guerreiro, and J. J. Sim{\' o}n, ``From ds-bounds for cyclic codes to true
minimum distance for Abelian codes", \emph{IEEE Trans. Inform. Theory}, vol. 65, no. 3, pp. 1752--1763, 2019.

\bibitem {C} V. Chepyzhov, ``A Gilbert-Vashamov bound for quasi-twisted codes of rate $1/n$", \emph{Proc. of the Joint Swedish-Russian Int. Workshop on Inf. Theory}, M\"olle, Sweden, pp. 214--218, 1993.

\bibitem {DH} R. Daskalov and P. Hristov, ``New quasi-twisted degenerate ternary linear codes", \emph{IEEE Trans. Inform. Theory}, vol. 49, no. 9, pp. 2259--2263, 2003.

\bibitem{EL} M. van Eupen and J. van Lint, ``On the minimum distance of ternary cyclic codes", \emph{IEEE Trans. Inform. Theory}, vol. 39, no. 2, pp. 409--422,  1993.

\bibitem{ELOT} M. F. Ezerman, S. Ling, B. \"Ozkaya and J. Tharnnukhroh, ``Spectral bounds for quasi-twisted codes", \emph{Proc. IEEE Int. Symp. Inf. Theory (ISIT)}, Jul. 2019, pp. 1922--1926. 

\bibitem{GO} C. G\"{u}neri and F. \"{O}zbudak, ``A bound on the minimum distance of quasi-cyclic codes", \emph{SIAM J. Discrete Math.}, vol. 26, no. 4, pp. 1781--1796, 2012.

\bibitem {HT} C. Hartmann and K. Tzeng, ``Generalizations of the BCH bound", \emph{Inf. Control}, vol. 20. no. 5, pp. 489--498, 1972.

\bibitem {H} A. Hocquenghem, ``Codes correcteurs d'Erreurs", \emph{Chiffres (Paris)}, vol. 2, pp. 147--156, 1959.

\bibitem {J} J. M. Jensen, ``The concatenated structure of cyclic and abelian codes", \emph{IEEE Trans. Inform. Theory}, vol. 31, no. 6, pp. 788--793, 1985.

\bibitem{Y} Y. Jia, ``On quasi-twisted codes over finite fields", \emph{Finite Fields Appl.}, vol. 18, pp. 237--257, 2012.

\bibitem{L2} K. Lally, ``Quasicyclic codes of index $\ell$ over $\Fq$ viewed as $\Fq[x]$-submodules of $\F_{q^{\ell}}[x]/ \langle x^m-1\rangle$", \emph{Proc. Conf. Appl. Algebra, Algebraic Algorithms and Error-Correcting Codes}, pp. 244--253, Springer, 2003.

\bibitem{LF} K. Lally and P. Fitzpatrick, ``Algebraic structure of quasi-cyclic codes", {\em Discrete Appl. Math.}, vol. 111, no. 1--2, pp. 157--175, 2001.

\bibitem{LW} J. van Lint and R. Wilson, ``On the minimum distance of cyclic codes", \emph{IEEE Trans. Inform. Theory}, vol. 32, no. 11, pp. 23--40,  1986.

\bibitem{LG} J. Lv and J. Gao, ``A minimum distance bound for 2-dimension $\lambda$-quasi-twisted codes over finite fields", \emph{Finite Fields Appl.}, vol. 51, pp. 146--167, 2018.

\bibitem{QSS} L. Qian, M. Shi, P. Sol\'{e}, ``On self-dual and LCD quasi-twisted codes of index two over a special chain ring", \emph{Crypt. and Comm., Disc. Struc., Bool. Func. and Seq.}, vol. 11, pp. 717--734, 2019.

\bibitem{RZ} D. Radkova and A. J. van Zanten, ``Constacyclic codes as invariant subspaces", \emph{Linear Alg. Appl.}, vol. 430, no. 2–3, pp. 855--864, 2009.

\bibitem {R2} C. Roos, ``A new lower bound for the minimum distance of a cyclic code", \emph{IEEE Trans. Inform. Theory}, vol. 29, no. 3, pp. 330--332, 1983.

\bibitem{SQS} M. Shi, L. Qian, P. Sol\'{e}, ``On self-dual negacirculant codes of index two and four", \emph{Des. Codes Crypto.}, vol. 11, pp. 2485--2494, 2018.

\bibitem{ST} P. Semenov and P. Trifonov, ``Spectral method for quasi-cyclic code analysis", \emph{IEEE Comm. Letters}, vol. 16, no. 11, pp. 1840--1843, 2012.

\bibitem{SZ} M. Shi and Y. Zhang, ``Quasi-twisted codes with constacyclic constituent codes", \emph{Finite Fields Appl.}, vol. 39, pp. 159--178, 2016.

\bibitem{Tanner} R. M. Tanner, ``A transform theory for a class of group-invariant codes", \emph{IEEE Trans. Inform. Theory}, vol. 34, no. 4, pp. 725--775, 1988.

\bibitem{LZ} A. Zeh and S. Ling, ``Decoding of quasi-cyclic codes up to a new lower bound on the minimum distance", \emph{Proc. IEEE Int. Symp. Inf. Theory (ISIT)}, Jun. 2014, pp. 2584--2588.

\bibitem{LZ2} A. Zeh and S. Ling, ``Spectral analysis of quasi-cyclic product codes", \emph{IEEE Trans. Inform. Theory}, vol. 62, no. 10, pp. 5359--5374,  2016.

\bibitem{WS} R. Wu and M. Shi, ``A modified Gilbert-Varshamov bound for self-dual quasi-twisted codes of index four", \emph{Finite Fields Appl.}, vol. 62, pp. 101627, 2020.

\end{thebibliography}
\end{document}